\documentclass[reqno]{amsart}
\usepackage{amssymb,mathtools,color,url}
\usepackage[utf8]{inputenc}
\usepackage{natbib} 
\usepackage{microtype}
\usepackage{enumitem}
\usepackage[hidelinks]{hyperref}

\newcommand{\ExerciseText}{Exercise}

\newcommand{\E}{\mathbb{E}}

\DeclareMathOperator{\Proj}{Proj}

\theoremstyle{plain}
\newtheorem{theorem}{Theorem}
\newtheorem{proposition}{Proposition}
\newtheorem{lemma}{Lemma}
\newtheorem{corollary}{Corollary}
\newtheorem{condition}{Condition}

\newtheorem{assumption}{Assumption}

\theoremstyle{definition}
\newtheorem{definition}{Definition}
\theoremstyle{remark}
\newtheorem{remark}{Remark}

\newtheorem{example}{Example}
\newtheorem{convention}{Convention}

\newenvironment{proof*}[1]{%
  \begin{proof}[#1]}{\end{proof}}
\newcounter{exercisenr} 

\newcommand{\remarkparagraph}[1]{\textbf{#1}\ }

\begin{document}
\renewcommand{\theenumi}{\roman{enumi}}%

\title{Short-horizon Duesenberry Equilibrium}
\thanks{This research did not receive any specific grant
  from funding agencies in the public, commercial, or not-for-profit sectors.
  Declarations of interest: none.}

\author{Jaime A. Londo\~no}
\address{Jaime A. Londo\~no, Departamento de Matem\'aticas y Estad\'\i{}stica,
  Facultad de Ciencias Exactas y Naturales,
  Universidad Nacional de Colombia, Sede Manizales,
  Manizales 170003, Colombia}
\email{jaime.a.londono@gmail.com}
\begin{abstract}
  We develop a continuous-time general equilibrium framework for an infinite
  heterogeneous population whose household types are transported by a Brownian
  flow. Agents repeatedly solve short-horizon optimization problems over
  state-price-discounted consumption and wealth, which admit an endogenous
  relative-income criterion in equilibrium. For the resulting Radner-type
  equilibrium---a short-horizon Duesenberry equilibrium---we obtain a complete
  characterization and an existence theorem; market completeness and the
  absence of state-tame arbitrage emerge endogenously from market clearing.
  Finally, the equity premium puzzle is shifted to a question about the
  volatility of aggregate total wealth, while the risk-free rate puzzle admits
  a plausible resolution.
\end{abstract}
\date{\today}
\maketitle
\tolerance=200 \setlength{\emergencystretch}{2em}

\section{Introduction}
\label{sec:introduction}

This paper investigates intertemporal equilibrium in
economies in which agents repeatedly solve short-horizon optimization problems
whose equilibrium representation is grounded in relative income.
The relative-income, short-horizon optimization mechanism proposed here breaks
the tight link between aggregate consumption volatility and asset returns. By
allowing preferences to move with demographics, the model can produce low
risk-free rates together with sizable equity premia, thereby offering a
unified framework that translates the classic puzzles into a question about
the volatility of aggregate total wealth---a quantity whose magnitude is
empirically plausible (see Remark~\ref{rem:equity_premium})---and provides a
plausible explanation of stylized facts related to population growth,
inflation, and interest rates (see
Remark~\ref{rem:population_growth_inflation}).

The present paper builds on two distinct components
of~\citet{Londono2020a}. First, the optimal consumption--investment result is
extended to an evolving continuum of heterogeneous household types. Second,
we develop the short-horizon mechanism suggested
in~\citet[Section~5.2]{Londono2020a} into an equilibrium model. In that paper, the
mechanism appears as a proposed extension, represented through a single stochastic
impatience process rather than a heterogeneous population field and analyzed by
means of a discretized limiting argument,
without a characterization or an existence theorem for the short-horizon
equilibrium.
Before introducing the technical apparatus, we describe the economy in words.
Each $x\in\mathbb{D}$, an open subset of $\mathbb{R}^{d}$, indexes a household
type. Households of the same type share a preference profile, and the mass
associated with each type may grow, shrink, or migrate over time. Households
receive labor-income endowment streams, trade a market portfolio and a bond, and
choose consumption and investment. Preferences are isoelastic, with an effective
impatience rate $\gamma(x)$ that might vary across types and evolves as the
cross-section ages and turns over. Equilibrium is a sequential-trading,
Radner-type notion: markets clear at every date, while households repeatedly
solve vanishing-horizon optimization problems.

The motivation for developing a new approach to equilibrium in markets based on
relative income arises from the misalignment between classical consumption
theories, equilibrium models of financial markets, and empirical data. This
misalignment is evidenced by several well-known puzzles, such as the ``equity
premium puzzle'' \citep{Mehra_and_Prescott1985}, the ``risk-free rate
puzzle'' \citep{Weil1989}, and the ``risk-aversion puzzle''
\citep{Jackwerth2000}. For a discussion of the theoretical problems with the
approaches to the existence of equilibrium,
see~\citet{Anderson_Raimondo2008},~\citet{Hugonnier_etal2010},~\citet{Kramkov_and_Prodoiu2012},~\citet{Herzberg_and_Reidel2012}
and~\citet{Raimondo2005}.

The use of utilities that compare one's own consumption to that of others is
not new; references date back to~\citet{Veblen1915} and~\citet{Duesenberry1949}
and include~\citet{Gali1994},~\citet{Abel1999},~\citet{Frank85},
\citet{Frank97},~\citet{Easterlin2002},~\citet{Scitovsky92},~\citet{Sen97},
and~\citet{Schor92} to cite a few. Closest to the present setting among these
is~\citet{ChanKogan2002}, who study a continuous-time exchange economy populated
by a continuum of agents with ``catching up with the Joneses'' preferences and
heterogeneous curvature. Evidence on relative-income effects in consumption and
investment is provided by~\citet{Easterlin1974},~\citet{Easterlin2002}
and~\citet{Easterlin2010}.

To address the misalignment described above, \citet{Londono2009, Londono2020a} proposed an
approach for optimal consumption and investment based on optimizing a
functional on consumption and wealth discounted by the state-price process. In
this paper, we generalize this solution to infinite populations, incorporating
population growth (see Theorem~\ref{thm:optimal_consumption_investment}).
\emph{A posteriori}, the optimal aggregate wealth and consumption can be
interpreted as a comparison of personal consumption with society's consumption
(or wealth), or with that of other consumers (see~\citealt[Remark~17]{Londono2020a}
and Remark~\ref{rem:keeping_up_joneses}). How this construction relates to the classical
keeping-up-with-the-Joneses formulations of~\citet{Gali1994}
and~\citet{Abel1999} is set out in Remark~\ref{rem:relation_joneses}.

To place our contribution in context, it is important to clarify the role of
the utility specification. At first sight, the preference structure resembles a
standard isoelastic formulation. The crucial difference, however, lies in its
argument. Following the valuation framework of~\citet{Londono2008} and the
state-dependent utility approach of~\citet{Londono2009}, utility is applied to
nominal consumption and wealth discounted by the state-price process, such as
$H_t c_t$, rather than to undiscounted consumption quantities $c_t$.

This formulation admits an endogenous relative-income representation in
equilibrium: although no social benchmark is introduced as a primitive argument
of utility, the local criterion can be written \emph{ex post} as utility from
individual consumption relative to aggregate total wealth, normalized by the
population-weighted effective aggregate $\eta_t$
(Corollary~\ref{cor:optimal_aggregate} and Remark~\ref{rem:keeping_up_joneses};
$\eta_t$ is deterministic under the fixed preference structure of
Theorem~\ref{thm:optimal_consumption_investment} and becomes stochastic in
Section~\ref{sec:equilibrium}, where the impatience field is transported by the
flow).

The solution obtained by optimizing relative wealth and consumption does not,
by itself, resolve the risk-free rate puzzle or the equity premium puzzle. To
address this shortcoming, \citet[Section~5.2]{Londono2020a} suggested an
approach---developed in this paper in Section~\ref{sec:equilibrium}---that
retains the optimal behavior described above and allows
consumers to solve their consumption and investment problem over short
horizons, and to reset their optimization continuously as preferences and
tastes change. This dynamic preference structure
may reflect changes in agents' tastes. These changes can arise from aging or
demographic shifts.  Since households never commit to a lifetime plan, the
resulting equilibrium path need not solve any single global relative-income
maximization problem.

As shown in Remark~\ref{rem:equity_premium}, this mechanism
provides an insight into both puzzles and suggests a resolution. The present study
develops and generalizes the proposed approach to a setting with infinitely
many heterogeneous agents and population growth, obtaining a solution to the
short-horizon optimization problem (see
Theorem~\ref{thm:local_optimal_consumption_investment}) and establishing the
characterization and existence of equilibrium under weak hypotheses (see
Theorem~\ref{thm:local_equilibrium}).

A two-type economy, developed in Remark~\ref{rem:two_type_example}, separates
two mechanisms that are easily conflated. The aggregate Duesenberry loading
$-\partial_t\eta_t$ is defined, under the fixed preference structure, in
Corollary~\ref{cor:optimal_aggregate}, and, under the dynamic preference
structure transported by the flow, in
Corollary~\ref{cor:optimal_local_aggregate} and in
Equation~\eqref{eq:varkappa_equilibrium} of
Theorem~\ref{thm:local_equilibrium}. When the effective impatience field is
locally constant, the relative-income representation is already available, yet
$-\partial_t\eta_t$ is deterministic and
carries no risk---and rolling the short-horizon problems forward does not change
this. What activates the impatience-risk premium of
Remark~\ref{rem:equity_premium} is not heterogeneity of
impatience as such, but its variation along the flow: once $\gamma$ is
non-constant on a connected type space, transport moves each type through
regions of differing impatience, $-\partial_t\eta_t$ acquires the diffusion
generated by $\nabla\gamma$, and its volatility enters the pricing kernel.

The contributions of the paper are the following.

\begin{enumerate}[label=(\roman*),leftmargin=*]

  \item \textbf{Optimal consumption and investment for infinite heterogeneous
          populations.} As a preliminary step toward our proposed approach, we provide an
        optimization result for the optimal consumption and investment (see
        Theorem~\ref{thm:optimal_consumption_investment}), which \emph{a posteriori}
        implies a maximization behavior on relative income (see
        Remark~\ref{rem:keeping_up_joneses}). The market is allowed to admit arbitrage opportunities and to be incomplete,
        and the population may be large and heterogeneous.
        Theorem~\ref{thm:optimal_consumption_investment} assumes that preferences
        remain unchanged as time evolves. This optimization result is a
        generalization to a large (heterogeneous) population
        of \citet[Theorem~11]{Londono2020a}.

  \item \textbf{Characterization and existence of the short-horizon Duesenberry
          equilibrium.} As a natural step to address equilibrium, we include a second
        group of results consisting of
        Theorem~\ref{thm:local_optimal_consumption_investment} and
        Theorem~\ref{thm:local_equilibrium}. From these two results, we derive
        consumption, investment, and portfolios as the limiting outcomes of the optimal
        processes over time intervals in which each agent solves the optimization
        problem of Theorem~\ref{thm:optimal_consumption_investment}. Finally, as a
        result of the aggregate consumption and wealth characterization, we obtain a
        characterization and existence theorem (Theorem~\ref{thm:local_equilibrium})
        for the implied intertemporal equilibrium under weak conditions.

  \item \textbf{Asset-pricing implications.}
        Corollary~\ref{cor:properties_local_equilibrium} synthesizes the macroeconomic
        consequences: the state-price process factorizes both as the ratio of the effective aggregate
        $\eta_t$ to aggregate total wealth and as the ratio of the aggregate Duesenberry
        loading $-\partial_t\eta_t$ to aggregate consumption; the latter is the direct
        analogue of the classical stochastic discount factor. Both the equity premium
        and the short interest rate admit explicit decompositions into a classical
        consumption-risk component and an impatience-risk
        component.  Remark~\ref{rem:equity_premium} shows that this
        decomposition identifies the source of the equity premium and risk-free
        rate puzzles: the market price of risk is the volatility of aggregate
        total wealth, and the values implied by long-run U.S.\ data cannot be
        reconciled with discounting labor income at the risk-free rate. The
        puzzles thus become a question about the rate at which households
        discount future consumption streams---plausibly an unsecured borrowing
        rate rather than $r_t$---which we take up in
        \citet{Londono2026w}. Remark~\ref{rem:population_growth_inflation}
        discusses qualitative implications for inflation pass-through and the
        secular decline of interest rates associated with demographic trends.

  \item \textbf{A dynamically consistent aggregation framework for evolving
          populations.}
        We develop a population-flow and aggregation structure under which
        heterogeneous household types can be transported over time, aggregates
        recomputed from any reset date coincide pathwise with their calendar-time
        counterparts, and population-weighted quantities remain semimartingales.
        Because aggregation is defined pathwise through the common Brownian flow, the
        particular measurability and exact-aggregation difficulties identified
        by~\citet{Judd1985} and addressed through the Fubini-extension framework
        of~\citet{Sun2006} do not arise in the present setting. The relative
        transparency of the resulting equilibrium formulas is therefore a consequence
        of the population-flow and consistency framework, rather than an indication
        that the underlying aggregation problem is immediate. The construction and its
        role in the equilibrium characterization are set out in
        Remark~\ref{rem:aggregation_framework}, with proofs collected in
        Appendix~\ref{appendix:aggregation}.

\end{enumerate}

Also, we review the leading approaches to the equity premium puzzle---habit
formation, long-run risk, and idiosyncratic-income heterogeneity---and argue
that each relies on at least one latent state variable or poorly identified
preference parameter to amplify the stochastic discount factor (see
Remark~\ref{rem:comparison_literature}).

We use the pricing theory based on state-tameness
\citep[see][]{Londono2004,Londono2008,Londono2020a}, which does not require
completeness or \emph{a priori} conditions on the volatility of the price
process---a feature that constitutes a major obstacle in the theory of the
existence of intertemporal equilibrium (see~\citealt{Londono2020a} and references
therein). A distinctive feature of the short-horizon Duesenberry equilibrium
framework is that market completeness (in the sense of no state-tame arbitrage
and well-defined pricing) emerges \emph{endogenously} from the equilibrium
conditions, rather than being imposed as an exogenous assumption. This contrasts
sharply with the classical literature \citep{Anderson_Raimondo2008, Hugonnier_etal2010,Kramkov_and_Prodoiu2012, Herzberg_and_Reidel2012}, where ``potentially dynamically complete
markets'' is a condition with no direct economic interpretation and must
therefore be assumed \emph{a priori}.

In our setting, the condition that implies completeness (the Smooth Market
Condition~\ref{cond:smooth_market}) is merely a regularity requirement
(continuity of $\kappa_t$), not a structural restriction on the market
architecture. The existence of $\kappa_t$ satisfying
$\kappa_t\sigma_t = \vartheta_t$, where $\sigma_t$ is the volatility of the
market price and $\vartheta_t$ is the market price of risk, is guaranteed for
any financial market; only its continuity is assumed. The no-arbitrage and
completeness property then follows as a \emph{consequence} of equilibrium, not
as a prerequisite for its existence. On the other hand, for the results studied
here, we only require that the income of any agent be hedgeable or insurable in
a rich enough actuarial market (see Remark~\ref{rem:hedgeable_income}), and we
do not require the hedgeability of any other financial instruments. Also, we
focus on the behavior of the aggregate market price rather than on that of
individual stocks. As shown by~\citet{Londono2020a}, Duesenberry equilibria do
not imply any particular behavior for individual securities. Instead, they
imply the absence of arbitrage opportunities when hedging with the market
portfolio and the bond.

To model changing preferences for relative consumption and wealth over time, we
define stochastic, consistent isoelastic preference structures for a
population. However, we point out that most definitions and results can be
extended to numerous types of utilities, including logarithmic utility
functions and homogeneous preference structures
\citep[see][Condition 1]{Londono2009}. We define stochastic dynamic preference
structures in two steps.

The first step is to define a \emph{time-consistent state-preference structure}:
a family of pairs of utilities for consumption and for wealth discounted by the
state-price process. Time consistency means that the same optimal consumption
and investment solution is obtained across all common investment horizons (see
Definition~\ref{def:consistent_preference}), allowing us  to address naturally
the problem of an infinitely lived agent.

The second step models the preference of each type of individual in the society
(represented by a point $x\in\mathbb{D}$) as a process with values in a family
of time-consistent preference structures $U_s$, whose realization
$U_s(\omega)(x)$ describes the preference of the agent at time~$s$ in
state~$\omega$. This stochastic dynamic preference structure captures the
possibility that individuals can change their preferences over time, for
instance, due to aging.

Some references on progressive, dynamic utilities to model changes over time of
preference behavior are~\citet{Musiela2007, Musiela2010}
and~\citet{Nicoleetal2018}, among others.

The economic model just described requires a state space rich enough to carry a
continuum of types together with their population dynamics, and this is what
dictates the mathematical apparatus of Section~\ref{sec:model}.
This work develops a framework to address state spaces characterizing finite
and infinite populations of heterogeneous consumers. We assume a state space
(an open set $\mathbb{D}\subset \mathbb{R}^d$), a Brownian flow of
two-parameter processes $\varphi_{s,t}(x)$ defined on
$\mathbb{D}\subset \mathbb{R}^d$ which models the stochastic behavior of the
state variables, and provides a natural framework for modeling infinite
populations. Indeed, our approach is an extension of the classical framework
where the relevant variables in the economy are of the form
$\int_U f(t,\varphi_t(x))\, d\mu(x)$ where $\varphi_t$ is an It\^o process
describing the evolution of state variables, $f$ is a differentiable function
describing the relation between the state variables and the quantities of
interest, and $\mu$ is the initial distribution on the state variables at time
$0$. In the context of mean field games and optimal portfolio management under
competition and relative performance
criteria,~\citet{LackerandZariphopolou2019} address the problem of mean field
games for portfolio optimization with infinite populations.
A separate literature studies portfolio choice and pricing under primitive
relative-performance concerns using stochastic games; see, among others,
\citet{EspinosaTouzi2015}, \citet{BielagkLionnetDosReis2017},
and~\citet{LackerSoret2020}.

The structure of the paper is as follows.

Section~\ref{sec:model} introduces the population dynamics, the pathwise
aggregation structure, and the financial market.

In Section~\ref{sec:utility}, we define consistent state preference structures
and stochastic consistent preference structures; we also extend a result
from~\citet{Londono2009, Londono2013, Londono2020a} on the optimal consumption
and investment theorem for utilities on state-price discounted values for
infinite populations.

In Section~\ref{sec:equilibrium}, we establish
Theorem~\ref{thm:local_optimal_consumption_investment}, which characterizes the
optimal consumption and investment strategy for consumers who continuously
reset their optimization problem, optimizing over state-price-discounted
consumption and wealth. Finally, we reach a characterization theorem on an
intertemporal equilibrium where agents use the locally optimal behavior of
Theorem~\ref{thm:local_optimal_consumption_investment} to determine their
consumption and investment decisions.

In Section~\ref{sec:examples}, we develop two examples that illustrate two
complementary parametrizations: one treating the present value of human wealth
as the primitive, the other starting from the labor-income process itself.

Finally, Section~\ref{sec:conclusions} presents several conclusions and
possible extensions of the framework proposed in this paper. Among the most
promising extensions are the application of the equilibrium framework to the
term structure of interest rates (fixed-income pricing), its generalization to
production economies in which the aggregate income process is endogenously
determined by production technology, and the development of \emph{Duesenberry
  equilibria with endogenous borrowing premia}, in which state-dependent
household borrowing rates generate large fluctuations in total wealth without
altering the aggregate pricing kernel.

The stochastic-flow preliminaries underlying the population dynamics are
collected in Appendix~\ref{appendix:flows}, while the cocycle and
semimartingale aggregation results are proved in
Appendix~\ref{appendix:aggregation}.

\section{The Economy}\label{sec:model}

Throughout, $(\Omega,\mathcal{F},(\mathcal{F}_t)_{t\ge0},P)$ is the canonical
filtered Wiener space with coordinate process $W$, and $(\mathcal{F}^T_{s,t})$
denotes the associated two-parameter filtration. Household types are indexed by
a point $x$ of an open set $\mathbb{D}\subset\mathbb{R}^{d}$, and their
evolution is described by a Brownian flow $\varphi_{s,t}(x)$ on $\mathbb{D}$: a
two-parameter family of continuous semimartingales, smooth in the spatial
variable, that transports each type forward in time. The construction of the flow, its regularity conditions, and the precise
notion of $\varphi$-consistency are collected in
Appendix~\ref{appendix:flows}.

\subsection{Population Dynamics and Aggregation}\label{subsec:population_dynamics}

The semimartingale properties of population-weighted aggregates and calendar-time consistency are fundamental to
the equilibrium analysis. We establish these properties in Appendix~\ref{appendix:aggregation},  under suitable
integrability conditions.  We
now introduce the population-flow and aggregation structure that underlies the equilibrium analysis.

\begin{definition}[Population-weighted aggregation]\label{def:aggregation}
  Let $\varphi$ be a Brownian flow satisfying Assumption~\ref{ass:Lipschitz},
  and let $h\colon\mathbb{D}\to\mathbb{R}$ be a deterministic population growth rate,
  bounded from above and below, of class $C^{2,0+}$.
  Let $\mathcal{P}_{\varphi}\subset\mathcal{P}(\mathbb{D})$ be a family of
  probability measures closed under transport by $\varphi$, that is, for every
  $0\le s\le t$ and every $\nu\in\mathcal{P}_{\varphi}$,
  \[
    \varphi_{s,t}(\nu)\in\mathcal{P}_{\varphi},
  \]
  where $\varphi_{s,t}(\nu)$ denotes the image of $\nu$ under
  $\varphi_{s,t}$ as in Definition~\ref{def:consistent_measure}.

  Define the population weight
  \[
    \Lambda_{s,t}(x)
    \coloneqq
    \exp\!\Big(\int_s^t h(\varphi_{s,u}(x))\,du\Big),
    \qquad 0\le s\le t,\; x\in\mathbb{D},
  \]
  and the measure-valued kernel $\nu_{s,t}(x,\cdot)\in\mathcal{M}(\mathbb{D})$ by
  \[
    \nu_{s,t}(x,B)
    \coloneqq
    \Lambda_{s,t}(x)\,\delta_{\varphi_{s,t}(x)}(B),
    \qquad B\in\mathcal{B}(\mathbb{D}).
  \]

  For any finite Borel measure $\mu$ on $\mathbb{D}$, define the action of the
  kernel $\nu_{s,t}$ on $\mu$ by
  \[
    (\mu\nu_{s,t})(B)
    \coloneqq
    \int_{\mathbb{D}} \nu_{s,t}(x,B)\,d\mu(x),
    \qquad B\in\mathcal{B}(\mathbb{D}).
  \]

  A \emph{population structure after time $s$} is a triple
  $\boldsymbol{\mu}=(f_{\mu},\nu_{s,t},\bar\mu_{s,t})$, where
  $f_{\mu}$ is an $\mathcal{F}_s$-measurable random variable such that
  $f_{\mu}\ge 0$ and $\E[|f_{\mu}|]<\infty$, and
  \[
    \mu \coloneqq f_{\mu}\,\bar\mu_{s,s},
    \qquad
    \bar\mu_{s,s}\in\mathcal{P}_{\varphi}(\mathbb{D}).
  \]

  We then define the structural no-growth composition by
  \[
    \bar\mu_{s,t}
    \coloneqq
    \varphi_{s,t}(\bar\mu_{s,s}),
    \qquad s\le t,
  \]
  and the associated size-weighted population measure by
  \[
    \mu_{s,t}
    \coloneqq
    \mu\nu_{s,t},
    \qquad\text{that is,}\qquad
    \mu_{s,t}(B)
    =
    \int_{\mathbb{D}}
    \mathbf{1}_{\{\varphi_{s,t}(x)\in B\}}\Lambda_{s,t}(x)\,d\mu(x).
  \]

  Assume $\bar\psi=(\bar\psi_{s,t}(x))$ is a continuous
  $C(\mathbb{D};\mathbb{R})$ semimartingale random field.
  The \emph{population-weighted aggregation of $\bar\psi$} is the
  $\varphi$-consistent family defined by
  \begin{equation}\label{eq:aggregation}
    \psi^\mu_{s,t}
    \coloneqq
    \int_{\mathbb{D}} \bar\psi_{s,t}(x)\,\Lambda_{s,t}(x)\,d\mu(x)
    =
    \int_{\mathbb{D}} \psi_{s,t}(x)\,d\mu(x),
  \end{equation}
  whenever the right-hand side is a well-defined semimartingale, where
  \[
    \psi_{s,t}(x)\coloneqq \bar\psi_{s,t}(x)\Lambda_{s,t}(x)
  \]
  is the mass-weighted typewise total field.
\end{definition}

\begin{convention}
  \label{conv:mu_notation}
  Throughout the paper, unless stated otherwise, we work with mass-weighted
  (typewise total) fields $X_{s,t}(x)$. Their per-capita counterparts are
  denoted by $\bar X_{s,t}(x)=X_{s,t}(x)/\Lambda_{s,t}(x)$.  For any typewise (per-capita) quantity $\bar{X}_{s,t}(x)$ or $\bar X_t(x)\coloneqq  \bar X_{0,t}(x)$,
  we write $X_{s,t}^\mu$ or $X_{t}^\mu\coloneqq X_{0,t}^\mu$ for its
  corresponding population-weighted aggregate. The same convention applies to the
  mass-weighted (typewise total) fields $X_{s,t}(x)$ and $X_t(x)\coloneqq X_{0,t}(x)$,
  whose population-weighted aggregates are denoted $X^\mu_{s,t}$ and $X^\mu_t$.
  This superscript notation is used only whenever a
  type-level counterpart exists, or when emphasizing the dependence on the
  population measure $\mu$; it may be omitted when the measure is clear from the context.
\end{convention}

In this paper, we focus on the behavior of the aggregate market price rather than on that of individual stocks. As shown by~\citet{Londono2020a}, Duesenberry equilibria do not imply any particular behavior for individual securities.

Since we assume that markets result from aggregated decisions, we require that the price processes and all their coefficients satisfy the fundamental  consistency condition given by the following definition (for a motivation, see Corollary~\ref{cor:aggregate_semimartingale_app}).

\begin{definition}\label{def:consistency_population}
  Assume  a population structure $\boldsymbol{\mu}$ with initial measure
  $\mu$ after time $0$ (see Definition~\ref{def:aggregation}).
  For each
  random measure  $\mu_t=\mu_{0,t}$, $t\geq 0$  we assume a  family of (random) positive continuous $\mathcal{F}_{t+s}$-adapted
  semimartingale processes $X^{\mu_t}=(X^{\mu_t}_{s})$ for $0\leq s< \infty$,
  with the property
  that
  \begin{equation}\label{eq:aggregate_consistency}
    X^{\mu_t}_u=X^{\mu}_{t+u}.
  \end{equation}
  We say that the family with the above properties is a
  \emph{family of processes consistent with the population structure}
  $\boldsymbol{\mu}$.  Note that a family as above is uniquely
  characterized by the process $X^{\mu}_t$ for $t\geq 0$.  In this
  case we use interchangeably the family $X^{\boldsymbol{\mu}}$ and the process $X_t^{\mu}$.

  Here $X^{\mu_t}_s$ denotes the evolution of the process $s$ units of time after date $t$,
  when the population aggregate is given by the random measure $\mu_t$.
  Throughout this paper, we adopt the notation
  $X^{\mu}_t$, for a family of processes satisfying this definition for
  a population structure $\boldsymbol{\mu}$, with initial measure $\mu$ at time $0$.
\end{definition}

\begin{remark}\label{rem:aggregation_framework}
  \remarkparagraph{Role of the aggregation framework.}
  The cocycle property of the population weights
  (Lemma~\ref{lem:aggregation_cocycle_app}) provides pathwise consistency across
  successive short-horizon problems, while the general aggregation result
  (Proposition~\ref{prop:semimartingaleProperty}) and its discrete and continuous
  specializations (Propositions~\ref{prop:disc_weighted_app}
  and~\ref{prop:cont_app}) provide the stochastic representations required by the
  equilibrium characterization of Section~\ref{sec:equilibrium}, in particular the
  semimartingale dynamics of the effective aggregate $\eta_t$ and of the aggregate
  Duesenberry loading $-\partial_t\eta_t$. These properties allow the sequence of
  locally optimal decisions to define a single intertemporal equilibrium process.

  The construction also avoids a device often used in models with a continuum of
  independently shocked agents. Aggregation here is pathwise: every household is transported by the
  same Brownian flow, heterogeneity enters only through the initial type
  $x\in\mathbb{D}$, and the realized profile $x\mapsto\varphi_{s,t}(x)$ is
  measurable by construction under Assumption~\ref{ass:Lipschitz}. Hence, unlike
  models that aggregate a continuum of independently shocked agents, the present
  framework does not require an exact law of large numbers or a Fubini-extension
  construction (see~\citet{Judd1985} and~\citet{Sun2006}). What the population structure contributes is heterogeneity in
  preferences and its transport over time, not idiosyncratic uncertainty.
\end{remark}

\subsection{Financial Market Structure}\label{subsec:financial_market_structure}

Next, we describe a financial market.
\begin{definition}[Financial market]\label{def:financial_market}
  Let $\boldsymbol{\mu}$ be a population structure with initial measure $\mu$
  after time $0$ satisfying Definition~\ref{def:aggregation}.

  A \emph{financial market} with population structure $\boldsymbol{\mu}$ is a structure
  $\mathfrak{M} = (P^\mu, \boldsymbol{\mu}, b^\mu, \sigma^\mu, \delta^\mu, \vartheta^\mu, r^\mu)$
  where each component is a continuous semimartingale process consistent with
  $\boldsymbol{\mu}$ in the sense of Definition~\ref{def:consistency_population},
  satisfying the following:

  \begin{enumerate}[label=(\roman*)]
    \item \textbf{Price process.} The \emph{market price process} $P^\mu = (P^\mu_t)_{t \geq 0}$
          is a positive continuous semimartingale with decomposition
          \begin{equation}\label{eq:evolution_market_price}
            dP^\mu_t = P^\mu_t \, b^\mu_t \, dt
            +P^\mu_t(\sigma^\mu_t)^\intercal\, dW(t)
            \qquad P^\mu_0 > 0,
          \end{equation}
          where the \emph{return process} $b^\mu = (b^\mu_t)$ and the \emph{volatility
            matrix (column vector) process} $\sigma^\mu = ((\sigma^\mu)^j_t)_{1 \leq j \leq n}$ are continuous
          processes consistent with $\boldsymbol{\mu}$, and $dW(t)=(dW_1(t),\cdots, dW_n(t))^\intercal$.

    \item \textbf{Dividend process.} The \emph{dividend yield process}
          $\delta^\mu = (\delta^\mu_t)$ is a continuous process consistent with $\boldsymbol{\mu}$.

    \item \textbf{Interest rate and bond.} The \emph{interest rate process}
          $r^\mu = (r^\mu_t)$ is a continuous process consistent with $\boldsymbol{\mu}$,
          and the \emph{bond price process} $B^\mu = (B^\mu_t)$ is the positive continuous
          semimartingale satisfying
          \begin{equation}\label{eq:Bond}
            dB^\mu_t = r^\mu_t \, B^\mu_t \, dt, \qquad B^\mu_0 = 1,
            \qquad t \geq 0.
          \end{equation}

    \item \textbf{Market price of risk.} The \emph{market price of risk process}
          $\vartheta^\mu = (\vartheta^\mu_t)$ is a continuous $\mathbb{R}^n$-valued (as column vector) process
          consistent with $\boldsymbol{\mu}$, with
          $\vartheta^\mu_t \in \big(\ker(\sigma^\mu_t)\big)^\perp$ for all $t \geq 0$,
          satisfying the no-arbitrage relation
          \begin{equation}\label{E:wviability}
            b^\mu_t + \delta^\mu_t - r^\mu_t
            - \Proj_{\ker(\sigma^\mu_t)^\intercal}\big(b^\mu_t + \delta^\mu_t - r^\mu_t\big)
            = (\sigma^\mu_t)^\intercal \, \vartheta^\mu_t.
          \end{equation}
  \end{enumerate}

  Given such a financial market $\mathfrak{M}$, the \emph{state-price process}
  $H^\mu = (H^\mu_t)$ is defined as the positive continuous semimartingale
  \begin{equation}\label{eq:state_price}
    H^\mu_t \coloneqq  \frac{Z^\mu_t}{B^\mu_t}, \qquad t \geq 0,
  \end{equation}
  where
  \begin{equation}\label{eq:density_process}
    Z^\mu_t \coloneqq  \exp\left\{
    -\int_0^t (\vartheta^\mu_u)^\intercal \, dW(u)
    - \frac{1}{2} \int_0^t \|\vartheta^\mu_u\|^2 \, du
    \right\}.
  \end{equation}
  By It\^o's Lemma, the inverse state-price process $(H^\mu_t)^{-1}$ satisfies
  \begin{equation}\label{eq:sde_consumptionprice}
    d(H^\mu_t)^{-1} = (H^\mu_t)^{-1} \big(r^\mu_t + \|\vartheta^\mu_t\|^2\big) \, dt
    + (H^\mu_t)^{-1} (\vartheta^\mu_t)^\intercal\, dW(t),
    \qquad (H^\mu_0)^{-1} = 1.
  \end{equation}
\end{definition}
Let us define $H_{s,t}^\mu=H^\mu_{t}/H^\mu_{s}$ for $s\leq t$. It follows
that  for $s\leq u\leq t$,  $H^\mu_{u,t}H_{s,u}^{\mu}=H_{s,t}^{\mu}$
and
\begin{equation*}
  d(H^{\mu}_{s,t})^{-1}=(H^{\mu}_{s,t})^{-1}\left(r^{\mu}_{t}+\|\vartheta^{\mu}_{t}\|^2\right)\, dt +(H^{\mu}_{s,t})^{-1}(\vartheta_t^\mu)^\intercal dW(t),\qquad (H^{\mu}_{s,s})^{-1}=1
\end{equation*}
for $s\leq t\leq T$, so $H^{\mu}_{s,t}$ for $s\leq t$ is the state-price process of the market reset at $s$.

Throughout this paper, we assume that the market satisfies Condition~\ref{cond:smooth_market} given below, which we call the smooth market condition.  Note that since
$\vartheta^\mu_{t}\in\ker^{\perp}(\sigma^\mu_{t})=\mathrm{Im}((\sigma^\mu_{t})^{\intercal})$
the existence of a progressively measurable function $\kappa$ with the property expressed in Equation~\eqref{eq:smooth_market} follows for any financial market \citep[see][]{Londono2004}.
Hence, the requirement below is a weak condition on the smoothness of the aforementioned property.

\begin{condition}[Smooth Market Condition]\label{cond:smooth_market}
  There exists a \emph{continuous} (scalar) semimartingale process $\kappa^{\mu}_{t}$
  consistent with $\boldsymbol{\mu}$ in the sense of Definition~\ref{def:consistency_population},
  taking values in $\mathrm{span}((\sigma^\mu)^1_{t},\cdots,(\sigma^\mu)_t^n)$,
  such that for all $t\ge 0$,
  \begin{equation}\label{eq:smooth_market}
    \kappa^\mu_{t}\sigma^\mu_{t}=\vartheta^\mu_{t}.
  \end{equation}
\end{condition}

\begin{remark}\label{rem:state_tameness}
  In general, the notion of Duesenberry equilibrium implies the dynamics of the market index (the whole market) and the dynamics of the  bond prices  (\citet[Theorem 16]{Londono2020a}).
  However, the Duesenberry equilibrium implies no restriction on the dynamics of individual stocks.
  Therefore, we assume that we can only trade bonds and the market index.

  We point out that unless
  $\Proj_{\ker((\sigma^\mu_{t})^{\intercal})}(b^\mu_{t}+\delta^\mu_{t}-r^\mu_{t})=0$, there are state-tame
  arbitrage opportunities.  Therefore, in this paper, we do not impose the absence of arbitrage at the aggregate level \emph{a priori}.
  See~\citet{Londono2004} for definitions and characterization of state arbitrage opportunities,
  and state-tame portfolios.
\end{remark}

Next, we review and extend some definitions from~\citet{Londono2009} needed to describe the Duesenberry equilibrium in the setting of this paper.
These extensions are natural adaptations of the classical theory of consumption and investment \citep{Karatzas1998}
to the new setting proposed herein. For a detailed description of
consistent and related processes, see~\citet{Londono2009} and~\citet{Londono2008}.

\begin{definition}\label{def:endowment}
  Assume a market $\mathfrak{M}$ with a population structure
  $\boldsymbol{\mu}$.
  Assume continuous, consistent with the population structure $\boldsymbol{\mu}$, real-valued random fields
  $\xi^{\boldsymbol{\mu}}\coloneqq (\xi^{\boldsymbol{\mu}}(x))$,
  $c^{\boldsymbol{\mu}}\coloneqq (c^{\boldsymbol{\mu}}(x))$,
  $Q^\mu\coloneqq (Q^\mu(x))$, ${\pi}^{\boldsymbol{\mu}}\coloneqq ({\pi}^{\boldsymbol{\mu}}(x))$
  and $L^{\boldsymbol{\mu}}\coloneqq (L^{\boldsymbol{\mu}}(x))$, with:
  \begin{description}
    \item[Rate of consumption, rate of endowment, wealth, portfolio, and hedgeability]
          Assume non-negative processes as above $c^\mu_{t}(x)$ and  $Q^\mu_{t}(x)$ such that  for all $x\in \mathbb{D}$,
          $\E\left[\int_{0}^{\infty}H^\mu_{t}c^\mu_{t}(x)\,dt\right]<\infty$ and $\E\left[\int_0^{\infty}H^\mu_{t}Q^\mu_{t}(x)\,dt\right]<\infty$.
          We say that a continuous random field $\pi_t^{\boldsymbol{\mu}}(x)$ is
          a portfolio that finances $\xi_t^\mu(x)$ with \emph{hedgeable rate of
            consumption and rate of endowment structure} if the family of processes
          $\xi^\mu_{t}(x)$ can be financed by $\pi^\mu_t(x)$ using the endowment $Q^\mu_{t}(x)$ and consumption $c^\mu_{t}(x)$, namely:
          \begin{multline}\label{eqn:financiability}
            (B^\mu_{t})^{-1}\xi^\mu_{t}(x)=\xi^\mu_{0}(x)+\int_{0}^t (B^\mu_{u})^{-1}
            \left(Q^\mu_{u}(x)-c^\mu_{u}(x)\right)\,du +\\
            \int_{0}^t(B^\mu_{u})^{-1}\pi^\mu_{u}(x)(\sigma^\mu_{u})^\intercal\,dW(u)
            +\int_0^t(B^\mu_{u})^{-1}\pi^\mu_{u}(x)(b^\mu_{u}+\delta^\mu_{u}-r^\mu_{u})\,du
          \end{multline}
          for all  $x\in\mathbb{D}$
          and $0\leq t<\infty$, and say that the family $(\xi^{\boldsymbol{\mu}},c^{\boldsymbol{\mu}}, Q^{\boldsymbol{\mu}})$ as above is a wealth, consumption, and income structure, where $\xi^{\boldsymbol{\mu}}$ is the \emph{wealth process structure},
          $c^{\boldsymbol{\mu}}$ is the \emph{rate of consumption structure}, and $Q^{\boldsymbol{\mu}}$
          is the \emph{rate of endowment structure}. Also,
          $\delta^\mu_{t}\pi^\mu_{t}(x) $ is the rate of dividend income accruing to the portfolio $\pi^\mu_t(x)$.

          We say that $\pi_t^\mu(x)$ is a no-arbitrage portfolio if for each $T>0$, and $0\leq t\leq T$,
          $H^\mu_{t}G^\mu_{t}(x)$  is uniformly bounded below,
          where the bound may depend on $x$ and $T$.  Here $G^\mu_{t}(x)$ denotes the  gain-in-excess process  \citep[see][Remark 1]{Londono2008} and it is defined by
          \begin{multline*}
            G^\mu_{t}(x)=B^\mu_{t}\int_{0}^t(B^\mu_{u})^{-1}\pi^\mu_{u}(x)(\sigma^\mu_{u})^\intercal\,dW(u)
            +\\
            B^\mu_{t}\int_0^t(B^\mu_{u})^{-1}\pi^\mu_{u}(x)(b^\mu_{u}+\delta^\mu_{u}-r^\mu_{u})\,du.
          \end{multline*}
          If $\pi_t^\mu(x)$ is a no-arbitrage portfolio, it follows \citep[see][Remark~6]{Londono2020a} that
          \begin{multline}\label{eqn:non_arbitrage_portfolio}
            H_t^\mu \xi_t^\mu(x) + \int_0^tH_u^\mu(c_u^\mu(x)-Q^\mu_u(x))\,du=\\
            \xi_0^\mu(x)+\int_0^tH^\mu_u\Big(\pi^\mu_u(x)(\sigma^\mu_u)^\intercal-\xi^\mu_u(x)(\vartheta^\mu_u)^\intercal\Big)\, dW(u).
          \end{multline}

          If a wealth process $(\xi^\mu,0,Q^\mu)$ is a
          hedgeable rate of consumption and rate of endowment structure,
          we say that $(\xi^\mu, Q^\mu)$ is a hedgeable
          wealth and rate of endowment structure.

    \item[Admissibility, and Subsistence random fields]
          A \emph{subsistence random field structure} $L^{\boldsymbol{\mu}}=(L^{\mu}_{t}(x))$
          $x\in\mathbb{D}$  for the market $\mathfrak{M}$
          is a wealth process structure where for each $T>0$, $L^\mu_{t}(x)H^\mu_{t}$,  $0\leq t\leq T$ is uniformly bounded
          below (where the bound may depend on $x$ and $T$) such that
          \[
            \E\left[H^\mu_{t}L^{\mu}_{t}(x)\right]<\infty
          \]
          for all $t$ and $x\in\mathbb{D}$.  We say that the couple
          $(\pi^{\boldsymbol{\mu}},c^{\boldsymbol{\mu}})$ of portfolio on stocks structure
          and rate of consumption structure, is \emph{admissible for $(L^{\boldsymbol{\mu}},Q^{\boldsymbol{\mu}})$},
          and write $(\pi^{\boldsymbol{\mu}},c^{\boldsymbol{\mu}})\in\mathcal{A}(L^{\boldsymbol{\mu}},Q^{\boldsymbol{\mu}})$ if
          \begin{equation*}
            \xi^{{\mu}}_{t}(x)\geq L^\mu_{t}(x)\qquad\mbox{for all $0\leq t<\infty$}, x\in\mathbb{D}.
          \end{equation*}
          This inequality means that financial wealth cannot fall below the (possibly negative) subsistence level, interpreted as the current value of future labor income (or its insured component).
  \end{description}
  The processes
  $c^{\mu}_t=\int_{\mathbb{D}}c^\mu_{t}(x)\,d\mu(x)$,
  $Q^{\mu}_t=\int_{\mathbb{D}}Q^\mu_{t}(x)\,d\mu(x)$,
  $\xi^{\mu}_{t}=\int_{\mathbb{D}}\xi^\mu_{t}(x)\,d\mu(x)$,
  $\pi^{\mu}_t=\int_{\mathbb{D}}\pi^\mu_{t}(x)\,d\mu(x)$,
  and
  $L^\mu_t=\int_{\mathbb{D}}L^\mu_{t}(x)\,d\mu(x)$
  are the \emph{aggregate rate of consumption process}, \emph{aggregate rate of endowment process},
  \emph{aggregate wealth process}, \emph{aggregate portfolio of stocks},  and  \emph{aggregate subsistence process},
  where  $\mu$ denotes the population measure at time $0$ introduced in
  Definition~\ref{def:aggregation}.
\end{definition}

The current value of future endowments is the present value of future labor income
(or of its insurable component) discounted by the state-price process.  We now define
the corresponding hedgeable structure, a key notion in what follows:

\begin{definition}\label{def:hedgeable_income}
  Assume a market $\mathfrak{M}$ with population structure
  $\boldsymbol{\mu}$.
  Assume an endowment process structure $Q^{\boldsymbol{\mu}}$, with the current value of
  future endowments structure $L^{\boldsymbol{\mu}}$
  \begin{multline*}
    L^\mu_{t}(x)=\frac{-1}{H^\mu_{t}}\E\left[\int_t^{\infty}H^\mu_{u}Q^\mu_{u}(x)\,du\mid\mathcal{F}_{t}\right]=\\
    \lim_{T\to\infty}\frac{-1}{H^\mu_{t}}\E\left[\int_t^{T}H^\mu_{u}Q^\mu_{u}(x)\,du\mid\mathcal{F}_{t}\right].
  \end{multline*}
  If $(L^{\boldsymbol{\mu}},Q^{\boldsymbol{\mu}})$ is a hedgeable wealth and rate of
  endowment structure with a portfolio  $\pi^{\boldsymbol{\mu}}_Q$,
  we say that $Q^{\boldsymbol{\mu}}$ is a \emph{hedgeable rate of endowment structure with current value of future endowments $L^{\boldsymbol{\mu}}$}. In fact,
  the portfolio $\pi_Q^\mu$ is indeed a no-arbitrage portfolio \citep[see][Theorem~11]{Londono2020a}.
\end{definition}

\begin{remark}\label{rem:hedgeable_income}
  In this paper, \emph{insurability} of labor income refers only to the \emph{insurable (non-macro)} components of labor-income risk.
  Specifically, we mean idiosyncratic mortality and disability risk, as well as idiosyncratic employment and wage risk,
  i.e.\ job-loss events and wage shocks that are \emph{not} driven by aggregate macroeconomic factors.
  Macroeconomic sources of labor-income variation (aggregate productivity, business-cycle conditions, etc.) are instead captured endogenously by the equilibrium model through traded aggregate risks and the state-price process, and therefore, they are
  not part of the ``insurable components'' discussed herein.

  As shown in~\citet{risks13050088}, in the presence of an actuarial market satisfying life-insurance completeness, mortality, and disability risks are hedgeable
  through insurance contracts.
  Similarly, unemployment or wage-insurance contracts may hedge idiosyncratic job-loss risk. In such markets, total wealth equals financial wealth plus the actuarially  priced value of insured labor income.

  Thus, we do not lose generality assuming hedgeability of labor income, as long as there is a sufficiently rich actuarial market in the sense of \citet{risks13050088}.
\end{remark}

\section{Optimal Consumption and Investment}\label{sec:utility}
Throughout this paper, we are mainly interested in portfolio evolution structures that we obtain as a result of the optimal behavior of consumers using isoelastic utilities:

\begin{definition}\label{def:utility_function}
  Consider a function $U\colon\left(0,\infty\right)\to\mathbb{R}$
  continuous, strictly increasing, strictly concave, and continuously
  differentiable, with $U^{\prime}(\infty)=\lim_{x\to\infty}U^{\prime}(x)=0$
  and $U^{\prime}(0+)\coloneqq  \lim_{x \downarrow 0}U^{\prime}(x)=\infty$.
  Such a function will be called a utility function.
\end{definition}

For every
utility function $U(\cdot)$, we denote by $\mathcal{I} (\cdot)$ the inverse of the
derivative $U^{\prime}(\cdot)$; both of these functions are continuous,
strictly decreasing and map $(0,\infty)$ onto itself with
$\mathcal{I} (0+)=U^{\prime}(0+)=\lim_{x\to 0^+}U^{\prime}(x)=\infty$,
$\mathcal{I} (\infty)=\lim_{x\to\infty}\mathcal{I} (x)=U^{\prime}(\infty)=0$. We extend $U$ by
$U(0)=U(0^+)$, and keep the same notation for the extension to $[0,\infty)$ of
$U$, trusting that it will be clear which function is meant.
It is a well-known result that
\begin{equation}
  \label{eq:duality}
  \max_{0<x<\infty}\left(U(x)-xy\right)=U(\mathcal{I} (y))-y\mathcal{I} (y),\qquad
  0<y<\infty.
\end{equation}

\begin{definition}\label{d:state_utility}
  Consider two continuous functions $U_1, U_2\colon [0,\infty)\times\left(0,\infty\right)\to\mathbb{R}$, such that $U_1(t,\cdot)$ and $U_2(t,\cdot)$
  are utility functions in the sense of Definition \ref{def:utility_function} for all
  $t\in [0,\infty)$.  Define $\mathcal{I}_1$ and $\mathcal{I}_2$   by  $\mathcal{I}_1(t,x)\coloneqq (\partial
    U_1(t,x)/\partial x)^{-1}$, and $\mathcal{I}_2(t,x)\coloneqq (\partial
    U_2(t,x)/\partial x)^{-1}$, where $(\partial U_1(t,\cdot)/\partial x)^{-1}$ and $(\partial U_2(t,\cdot)/\partial x)^{-1}$
  denote the inverses of the derivatives of $U_1(t,\cdot)$
  and $U_2(t,\cdot)$, respectively;  it follows that $\mathcal{I}_1$ and $\mathcal{I}_2$ are
  continuous functions.  We call the pair $(U_1,U_2)$ a \emph{state preference structure}.  In addition, we assume the following notation for any preference structure:
  \begin{equation}\label{eq:inti}
    \mathcal{X}(t,T,y)\coloneqq  \mathcal{I}_2(T,y)+\int_t^{T}\mathcal{I}_1(s,y)\,ds
  \end{equation}
  for $0\leq t\leq T<\infty$.
\end{definition}

We extend $U_1$ and $U_2$ by defining $U_1(t,0)=U_1(t,0^+)$,  and $U_2(t,0)=U_2(t,0^+)$ for all $0\leq t
  \leq \infty$, and we keep the same notation for the extension
of $U_1$ and $U_2$ to $[0,\infty)\times [0,\infty)$.  These utility functions
are discussed in~\citet{Londono2009}.  We observe that the following interpretation
is natural. $U_1(t,\cdot)$ is the utility for consuming $c_t$ dollars
(as seen at time $0$) discounted by the state price at time $t$, and
$U_2(t,\cdot)$ is the utility of wealth discounted by the
state price (as seen at time $0$) when premature death occurs at time $t$.

Natural classes of utility structures are those whose utility for terminal wealth corresponds to the utility of future consumption that is not realized
because of premature death.  The precise definition is as follows:

\begin{definition}\label{def:consistent_preference}
  A state preference structure is \emph{time consistent} if for any $0\leq  T^{\prime}\leq T$
  \begin{equation*}
    \mathcal{I}_2(T^{\prime},y)-\mathcal{I}_2(T,y)=\int_{T^{\prime}}^T\mathcal{I}_1(s,y)\,ds.
  \end{equation*}
  Moreover, we say that the state preference structure has
  \emph{integrable  inverse marginal utility} if, for each $y$,
  $\sup_{0\leq t<\infty}\mathcal{I}_2(t,y)\vee \int_0^{\infty}\mathcal{I}_1(t,y)\,dt < \infty$.
\end{definition}

If $(U_1, U_2)$ is a time-consistent state preference structure, the functions $\mathcal{X}(t,T,y)$ and $\mathcal{X}^{-1}(t,T,y)$ do not depend on $T$,
and we denote them by $\mathcal{X}(t,y)$ and $\mathcal{X}^{-1}(t,y)$.

\begin{remark}\label{rem:epz}
  We choose consistent isoelastic preference structures because they are scale-invariant and allow for aggregation.
  For each isoelastic preference structure, the choice of $U_2(T,y)$ implies that the solution for consumption obtained by Theorem~\ref{thm:optimal_consumption_investment} for $0\leq t \leq T^{\prime}$  with $T<T^{\prime}$ agrees with the corresponding solution for horizon $T$ on $0\leq t\leq T$, implying that $U_2(T,y)$, the utility of final wealth $y$ discounted by the state price, can be seen as the utility for future consumption.  In this way, the utility structure suggests a utility preference
  structure for a population $U_1^{\infty}(x)$ of an infinitely lived
  agent characterized by $x\in\mathbb{D}$, in a framework similar to
  that used to describe the equity premium  and interest-rate puzzles.

\end{remark}

\begin{example}\label{ex:utility_structure1}
  Let $c>0$, $0<\alpha<1$, and  $0<\beta<1$.  Define  $U_1(t,y)=ce^{-\beta t}y^{\alpha}$ and $U_2(t,y)=de^{-\beta t}y^{\alpha}$, where $d=c((1-\alpha)/\beta)^{1-\alpha}$.  It is straightforward to see
  that $U_1, U_2\colon [0,\infty)\times (0,\infty)\to (0,\infty)$ define a time-consistent preference structure with integrable inverse marginal utility.
\end{example}

\begin{definition}\label{def:utility}
  We say that  a consistent isoelastic preference structure    is a family of utility functions of the form
  \[
    U_1(t,y)=ce^{-\beta t}y^{\alpha}, \qquad U_2(t,y)=c\left(\frac{1-\alpha}{\beta}\right)^{1-\alpha}e^{-\beta t}y^{\alpha},
  \]
  for $0\leq t<\infty$, $0<\alpha<1$, $1>\beta>0$, and $c>0$.
  A preference structure  for a population is defined as a parametric
  family of consistent isoelastic preference structures
  $U(x)=(U_1(x), U_2(x))$ where $U_1(x)=(U_1(t,\cdot)(x)\colon 0\leq t)$
  and $U_2(x)=(U_2(t,\cdot)(x)\colon 0\leq t )$ and the pair  $(U_1(x),U_2(x))$ is a state preference
  structure for each $x\in\mathbb{D}$.  Moreover, we say that the
  preference structure for a population is a \emph{consistent isoelastic preference
    structure} if for each $x$, $U(x)$ is a consistent isoelastic preference structure
  where we assume that the \emph{effective impatience rate} $\gamma(x)=\beta(x)/(1-\alpha(x))$
  is a smooth function of $x$, where $\alpha(x)$ and $\beta(x)$ correspond
  to the exponent parameter and time discount factor, respectively, for the isoelastic preference structure $(U_1(x), U_2(x))$.  The interpretation is the following:
  the preference structure $U(x)$ represents the
  preference behavior toward intermediate consumption and wealth for an
  agent characterized by the state value $x$ for future consumption.

  An \emph{isoelastic  consistent stochastic  preference structure
    for a population} is a process
  $U_{s,t}\colon \Omega\to \mathcal{U}$ for
  $0\leq s\leq t$ where $\mathcal{U}$ is a class of isoelastic
  preference structures for a population,  where  the corresponding
  \emph{effective impatience field}  $\gamma_{s,t}(\omega)(x)$ of
  $U_{s,t}(\omega)(x)$, is   $\gamma(\varphi_{s,t}(x)(\omega))$
  which is a  continuous random field, where  $\gamma(\cdot)$ is a
  positive function of type $C^{2}(\mathbb{D})$. We say that
  \emph{the generator of the isoelastic  consistent stochastic
    preference structure  for a population}  is the family
  $U_t=U_{0,t}$ with \emph{effective impatience rate} $\gamma(\varphi_{t}(x)(\omega))$. To avoid any unnecessary technicality, we do not assume any
  measurability of $U_{s,t}$ but instead, we assume continuity of $\gamma$.
\end{definition}

A natural interpretation of an isoelastic stochastic preference
structure $U_{s,t}$ is as follows.  An agent is characterized
by a type $x\in\mathbb{D}$ at time $s$, whose state value evolves
according to $\varphi_{s,t}$.  Moreover, the value of utility for
consumption and (final) wealth after time $t$, given that at time $s$
its state value is $x$, is the isoelastic preference structure with
effective impatience  rate $\gamma(\varphi_{s,t}(x)(\omega))$.

To study the optimal aggregate behavior of a population with a time-homogeneous population growth rate, we need to impose some conditions on the population structure:

\begin{assumption}\label{ass:population_preference_structure}
  $\boldsymbol{\mu}$ is a population structure with a deterministic population
  growth rate $h$,  as in Definition~\ref{def:aggregation}.  Assume consumers have a  consistent isoelastic preference structure
  for a population $U(x)$ for $x\in \mathbb{D}$, where
  $U(x)=(U_1(x),U_2(x))$. Assume that the effective impatience rate $\gamma(x)$ is a \emph{bounded} (from below and above) $C^{2}(\mathbb{D})$
  function satisfying $\gamma(x)\geq \underline{\gamma}$
  for some $\underline{\gamma}>0$.
\end{assumption}

We present below a  theorem that describes the optimal behavior of consumers (Theorem~\ref{thm:optimal_consumption_investment}).
In the following theorem, which assumes no changes in tastes, we do not assume that the market is free of (state) arbitrage or (state) complete.

\begin{theorem}\label{thm:optimal_consumption_investment}
  Assume a population structure $\boldsymbol{\mu}$  after time $0$ as in
  Definition~\ref{def:aggregation}, satisfying
  Assumption~\ref{ass:population_preference_structure}.  Let
  $\mathfrak{M}=(P,\boldsymbol{\mu}, b,\sigma,\delta,\vartheta,r)$ be a financial
  market that satisfies the Smooth Market Condition
  (Condition~\ref{cond:smooth_market}), where $\kappa_{t}$ is the (continuous)
  process  taking values  in $\mathrm{span}((\sigma)^1_{t},\cdots,(\sigma)_t^n)$ with $\kappa_{t}\sigma_{t}=\vartheta_{t}$.

  Also, assume consumers with a consistent isoelastic preference structure
  for a population,  $U(x)$ for $x\in \mathbb{D}$ with respect to $\varphi_t(x)$ that encodes the random
  dynamic of all state variables, where
  $U(x)=(U_1(x),U_2(x))$ and  $C^{2}(\mathbb{D})$ parameter  function
  $\gamma(x)$ where $\gamma(\varphi_{t}(x))=\gamma(x)$   is constant along the
  evolution of the flow.  In other words, $\varphi_t(x)$ only encodes the population dynamics,
  and on each type $x$, the preference remains the same at all times. Assume a hedgeable rate of endowment (population-scaled) structure
  $Q_t(x)$ (with hedging portfolio on
  the stocks $\varpi_{t}(x)$, which is itself a no-arbitrage portfolio),
  and the current value of
  (population-scaled) future endowments $L_t(x)$, as in
  Definition~\ref{def:hedgeable_income}.

  Let $\xi$ be defined by
  \[\xi_{t}(x)\coloneqq
    L_{t}(x)+ e^{-t\gamma(x)}H_{t}^{-1}(\xi_{0}(x)-L_{0}(x))
  \]
  and $c$ defined by
  \[
    c_{t}(x)\coloneqq
    \gamma(x) e^{-t\gamma(x)}H_{t}^{-1}(\xi_{0}(x)-L_{0}(x)),
  \]
  with a portfolio of stocks  given by
  \begin{equation*}
    \pi_{t}(x) \coloneqq e^{-t\gamma(x)}H_{t}^{-1}(\xi_{0}(x)-L_{0}(x))\kappa_{t}-\varpi_{t}(x)
  \end{equation*}
  for each $t\geq 0$. Then, $(\xi_{t},{c}_{t},{Q}_{t})$ is a hedgeable rate of  consumption and rate of endowment
  structure and no-arbitrage portfolio
  $(\pi_{t},{c}_{t})\in\mathcal{A}({L}_{t},{Q}_{t})$ that is optimal for the problem of optimal
  consumption and investment, where ${L}_{t}$ is the population-scaled current value of future endowments.
  The optimality is in the sense that for all $T>0$,
  \begin{multline}\label{eq:utility_comparison}
    \E[\int_{0}^{T}U_1(t,H_{t}{c}_{t}(x))(x)\,
      d t
      + U_2(T,H_{T}{\xi}_{T}(x))(x) ]\geq\\
    \E[\int_{0}^{T}U_1(t,H_{t}\tilde{c}_{t}(x))(x)\,
      dt + U_2(T,H_{T}\tilde{\xi}_{T}(x))(x) ]
  \end{multline}
  for all  $x\in\mathbb{D}$
  where $(\tilde{\xi}_{t},\tilde{c}_{t},\tilde{Q}_{t})$ is any other population-scaled hedgeable rate of
  consumption and rate of endowment structure.
\end{theorem}
\begin{proof}
  If $({\xi},{c},{Q})$ is defined as above, it follows that
  Equation~\eqref{eqn:non_arbitrage_portfolio} holds.
  Using It\^o's Lemma it follows that ${\xi}(x)$ can be financed
  using the endowment  ${Q}_{t}(x)$ and consumption ${c}_{t}(x)$
  (see Equation~\eqref{eqn:financiability}), since $\varpi_{t}(x)$
  is a no-arbitrage portfolio for $Q_t$.     The proof follows the lines of the proof of~\citet[Theorem~11]{Londono2020a}
  and~\citet[Theorem~2]{Londono2009}
  with the appropriate modifications.  We emphasize
  that the smooth market condition (Condition~\ref{cond:smooth_market})
  is used in the proof of this theorem.
\end{proof}

\begin{corollary}
  \label{cor:optimal_aggregate}
  Assume the conditions of Theorem~\ref{thm:optimal_consumption_investment}
  and assume Assumption~\ref{ass:population_preference_structure}.
  Assume that  $(\xi, c,Q)$  is a wealth, consumption, and income structure.  We define the
  \emph{population-weighted effective aggregate process} $\eta_{t}=\eta^\mu_{t}$ as follows:
  \begin{equation*}
    \eta_{t}
    \coloneqq  \int_{\mathbb{D}} e^{-t\gamma(x)}
    \big(\xi_{0}(x)-L_{0}(x)\big)\,d\mu(x),
    \qquad t\ge 0.
  \end{equation*}
  Then the (population-weighted) aggregate optimal wealth process satisfies
  \begin{equation*}
    \xi^\mu_{t} - L^\mu_t
    \coloneqq  \int_{\mathbb{D}}
    \big(\xi_t(x)- L_t(x)\big)\,d\mu(x)
    = \frac{1}{H_{t}}\eta_{t},
  \end{equation*}
  and
  \begin{equation*}
    c^\mu_{t}
    \coloneqq  \int_{\mathbb{D}} c_t(x)\,d\mu(x)
    = \frac{1}{H_{t}}\left(-\partial_t\eta_t\right),
  \end{equation*}
  for all $t>0$, where
  \begin{equation*}
    -\partial_t\eta_t=
    \int_{\mathbb{D}}
    \gamma(x)\,e^{-t\gamma(x)}
    \big(\xi_{0}(x)-L_{0}(x)\big)\,d\mu(x),\qquad t\ge0,
  \end{equation*}
  where $L^\mu_t$ is the aggregate of the current value of
  future endowments   $L_t(x)$, namely
  \[
    L^\mu_t = \int_{\mathbb{D}} L_t(x)\, d\mu(x).
  \]
  Moreover, the aggregate  optimal portfolio in stocks
  is given by
  \begin{equation*}
    \pi_{t}^\mu
    \coloneqq  \int_{\mathbb{D}} \pi_t(x)\,d\mu(x)
    = (H_{t})^{-1}\eta_{t}\,
    \kappa_{t}
    - \varpi^\mu_{t},
  \end{equation*}
  where $\varpi^\mu_{t}$ is the aggregate hedging portfolio process
  associated with the  hedging  portfolio $\varpi_{t}(x)$ (of the endowment
  $Q_t(x)$), \(\varpi^\mu_t=\int_{\mathbb{D}}\varpi_t(x)\,d\mu(x)\).
\end{corollary}

\begin{remark}\label{rem:keeping_up_joneses}
  The criterion in Theorem~\ref{thm:optimal_consumption_investment} offers conceptual simplicity, yet a natural objection is that individual agents cannot directly observe the financial state-price deflator $H_t$. This apparent difficulty dissolves once we examine the \emph{ex post} structure of optimal behavior.

  From Theorem~\ref{thm:optimal_consumption_investment} and Corollary~\ref{cor:optimal_aggregate},
  the comparison criterion~\eqref{eq:utility_comparison} can be represented as follows:
  For any reference type $y\in\mathbb{D}$, define its relative total-wealth gain
  \[
    G_t(y)\coloneqq \frac{{\xi}_t(y)-{L}_t(y)}{{\xi}_0(y)-{L}_0(y)}.
  \]
  Then the optimization criterion used in Theorem~\ref{thm:optimal_consumption_investment}
  is based on the maximization of
  \begin{multline*}
    E\Big[\int_{0}^{T}U_1\big(t,H_{t}{c}_{t}(x)\big)(x)\,dt
      + U_2\big(T,H_{T}{\xi}_{T}(x)\big)(x)\Big]\\
    =E\Big[\int_{0}^{T}U_1\Big(t,{c}_{t}(x)\frac{e^{-t\gamma(y)}}{G_t(y)}\Big)(x)\,dt + U_2\Big(T,{\xi}_{T}(x)
      \frac{e^{-T\gamma(y)}}{G_T(y)}\Big)(x)\Big]\\
    =E\Big[\int_{0}^{T}U_1\Big(t,{c}_{t}(x)\frac{\eta_t}{\xi^\mu_{t} - L^\mu_t}\Big)(x)\,dt
      + U_2\Big(T,{\xi}_{T}(x)\frac{\eta_T}{\xi^\mu_{T} - L^\mu_T}\Big)(x)\Big],
  \end{multline*}
  where $\eta_t\coloneqq \int_{\mathbb{D}}e^{-t\gamma(x)}({\xi}_0(x)-{L}_0(x))\,d\mu(x)$,
  and the parameter $\gamma$  encodes impatience.

  This representation reveals a fundamental insight: \emph{ex post}, agents are not maximizing wealth in the traditional sense but rather maximizing utility over relative income---specifically,
  their consumption and terminal wealth relative to the population-weighted total wealth
  $(\xi_t^\mu - L_t^\mu)$. The peer-group representation
  (first identity) shows equivalently that agents compare themselves against the wealth growth
  $G_t(y)$ of a reference cohort. This relative income maximization---rather than absolute
  wealth maximization---is the defining characteristic that justifies calling our framework a
  \emph{Duesenberry equilibrium}, connecting it directly to \citet{Duesenberry1949}'s relative
  income hypothesis.

  From a practical perspective, a typical agent does not need to observe $H_t$ directly. Instead, by
  tracking the wealth dynamics of a reference peer group---information that is plausibly
  available through ordinary social and economic observation---the agent obtains sufficient
  information to make consumption and investment decisions over short horizons. The ratio
  $G_t(y)$ summarizes the growth of total wealth (financial plus human) for the reference cohort
  and delivers an operational proxy for $H_t$ through relative comparisons.
  The last identity confirms that, in equilibrium, this peer-group proxy coincides with the
  aggregate benchmark.

  It is important to emphasize that this relative comparison criterion operates \emph{only}
  for short-horizon decisions. Agents use peer comparisons to determine their optimal
  consumption and investment over each short interval; the full dynamic solution
  emerges by rolling forward these local decisions, as we explain in Definition~\ref{def:structure_partition}, Theorem~\ref{thm:local_optimal_consumption_investment}, and Theorem~\ref{thm:local_equilibrium}.
  Beyond the short-horizon optimization step, the interpretation of ``keeping up with the Joneses'' no longer applies---the rolling procedure is purely mechanical.
\end{remark}

\begin{remark}\label{rem:relation_joneses}
  \remarkparagraph{Relation to keeping-up-with-the-Joneses preferences.}
  Classical formulations such as~\citet{Gali1994} and~\citet{Abel1999} introduce a
  reference-consumption process directly into the household's primitive
  intertemporal utility specification. The atomistic household takes this
  benchmark as given, while an equilibrium consistency condition relates it to
  aggregate consumption.

  Our construction is analogous, but differs in the status, the nature, and the
  scope of the comparison. No reference process is imposed as a primitive argument
  of household utility. Instead, Corollary~\ref{cor:optimal_aggregate} identifies
  the equilibrium state-price deflator as
  \[
    H_t=\frac{\eta_t}{\xi_t^\mu-L_t^\mu},
  \]
  where $\xi_t^\mu-L_t^\mu$ is aggregate total wealth, including financial wealth
  and human capital, and $\eta_t$ is the population-weighted effective aggregate
  process. Hence, defining
  \[
    R_t\coloneqq \frac{\xi_t^\mu-L_t^\mu}{\eta_t},
  \]
  the local felicity argument can be written as
  \[
    H_t c_t(x)=\frac{c_t(x)}{R_t}.
  \]
  Thus, each short-horizon criterion admits an endogenous relative-income
  representation in which the benchmark is aggregate total wealth, normalized by
  the effective aggregate $\eta_t$. The benchmark enters through the ratio $c_t(x)/R_t$, with no
  additional comparison-intensity parameter.

  The scope of the comparison is delimited as follows. Under a fixed preference
  structure the relative-income criterion is the household's genuine
  intertemporal objective: the strategy of
  Theorem~\ref{thm:optimal_consumption_investment} is optimal simultaneously for
  every horizon $T>0$, by the time consistency of
  Definition~\ref{def:consistent_preference}. It is only in the equilibrium
  results of Section~\ref{sec:equilibrium}, where the preference structure
  itself evolves, that the criterion becomes auxiliary: there each household
  re-solves the problem over a vanishing horizon, and the equilibrium path is
  generated by rolling those local decisions forward, so that it need not
  maximize any relative-income functional over the whole horizon.

  To make the comparison channel explicit, consider the isoelastic local felicity
  \[
    \widetilde U_t(c,\mathcal{W};x)
    =
    K_t(x)\left(\frac{\eta_t c}{\mathcal{W}}\right)^{\alpha(x)},
    \qquad
    \mathcal{W}>0,
    \quad
    K_t(x)>0,
    \quad
    0<\alpha(x)<1,
  \]
  where $\mathcal{W}$ denotes aggregate total wealth. Holding $K_t(x)$ and $\eta_t$
  fixed, as is appropriate for an atomistic household taking the aggregate
  benchmark as given, we obtain
  \[
    \frac{\partial \widetilde U_t}{\partial c}
    =
    K_t(x)\,\alpha(x)\,\eta_t^{\alpha(x)}\,
    c^{\alpha(x)-1}\,\mathcal{W}^{-\alpha(x)}>0,
  \]
  and
  \[
    \frac{\partial^2\widetilde U_t}{\partial c\,\partial\mathcal{W}}
    =
    -K_t(x)\,\alpha(x)^2\,\eta_t^{\alpha(x)}\,
    c^{\alpha(x)-1}\,\mathcal{W}^{-\alpha(x)-1}<0.
  \]
  Therefore, at a given date and state, an increase in the aggregate-wealth
  benchmark reduces both relative felicity and the marginal utility of own
  consumption. This is the same negative relative-position effect emphasized
  in~\citet[Remark~18]{Londono2020a}, now expressed with aggregate total wealth
  rather than aggregate consumption as the equilibrium benchmark. The restriction
  $0<\alpha(x)<1$ additionally guarantees that local felicity is increasing and
  strictly concave in own consumption.
\end{remark}

\begin{remark}\label{rem:type_level_feasibility}
  A potential source of confusion is that decisions are made by individual households (or families), not by population aggregates. The point is that under
  CRRA (isoelastic) preferences and infinite divisibility, a household with
  time-varying mass can be represented by a \emph{type block} whose natural quantities are mass-scaled.

  For each type $x$, writing $\Lambda_t(x)\coloneqq \Lambda_{0,t}(x)$ for the population
  weight of Definition~\ref{def:aggregation}, we consider per-capita wealth,
  consumption, and endowment $\bar\xi_t(x)\coloneqq \Lambda^{-1}_t(x)\xi_t(x)$,
  $\bar c_t(x)\coloneqq \Lambda^{-1}_t(x)c_t(x)$, $\bar Q_t(x)\coloneqq \Lambda^{-1}_t(x)Q_t(x)$,
  $\bar\pi_t(x)\coloneqq \Lambda^{-1}_t(x)\pi_t(x)$. Population growth changes $\Lambda_t(x)$ by introducing new mass (new entrants)
  and hence new aggregate resources; it is therefore \emph{not} economically
  meaningful to require a per-capita self-financing identity for
  $(\bar\xi_t(x),\bar c_t(x),\bar Q_t(x))$, since this would implicitly treat the inflow due to
  mass growth as a cash flow replicable by trading. Feasibility/admissibility is
  instead imposed on the block-level processes $(\xi,c,Q)$, for which the standard deflated budget identity holds under the market deflator $H$.

  Equivalently, one may view the model through a short-horizon (rolling) lens:
  over a small interval $[t,t+\Delta]$ the mass is approximately constant, so each
  household solves a standard consumption--investment problem; the update from
  $t$ to $t+\Delta$ then incorporates the net entry of mass at rate $h$, and the continuous-time limit yields precisely the block-level feasibility used in
  Theorem~\ref{thm:optimal_consumption_investment}.
\end{remark}

\section{Short-horizon Duesenberry Equilibrium: Characterization and Existence}
\label{sec:equilibrium}

Next, we describe the setting of the model we propose for the equilibrium.

\begin{definition}\label{def:structure_partition}
  Assume a financial market  $\mathfrak{M}=(P, \boldsymbol{\mu}, b,\sigma,\delta,\vartheta,r)$ that satisfies the Smooth Market Condition (Condition~\ref{cond:smooth_market}).
  Suppose an isoelastic consistent  preference structure for a
  population $U$ with corresponding parameter   $\gamma(\varphi_t(x)(\omega))$
  of $U_t(\omega)(x)$, where  $\gamma(\cdot)$ is a function of
  type $C^{2}(\mathbb{D})$, satisfying  Assumption~\ref{ass:population_preference_structure}.
  Assume an aggregate hedgeable (by state-tame portfolios) rate of consumption and rate of endowment structure
  $(\xi,c, Q)$  and no-arbitrage portfolio  $\pi$.
  Let $\Gamma=\{s_0=0<s_1<\cdots <s_m=T\}$ be a partition of the interval
  $[0, T]$.
  The \emph{endowment and consumption structure associated with
    $\Gamma$ and  preference structure $U$} satisfying
  Assumption~\ref{ass:population_preference_structure},
  is the pair of endowment and consumption processes  defined by the wealth-consumption-and-income structure $(\xi^{\Gamma},c^{\Gamma}, Q)$ explained
  below where  $\xi^{\Gamma}=(\xi^{\Gamma}_{t}(x))$ is a continuous wealth
  process  and $c^{\Gamma}=(c^{\Gamma}_{t}(x))$
  is the piecewise continuous process, where  on each interval $[s_i,s_{i+1}]$
  it coincides with the optimal process described in Theorem~\ref{thm:optimal_consumption_investment}.
  In other words, for any $x\in\mathbb{D}$ and $s_{k}\leq t< s_{k+1}$,
  writing $s_k^+=t$ and $s_i^+=s_{i+1}$ for $0\leq i\leq k-1$, we have:
  \begin{multline}\label{eq:discrete_wealth}
    \xi^{\Gamma}_{t}(x)-
    L_{t}(x)\coloneqq
    e^{-(t-s_k)\gamma(\varphi_{s_k}(x))}H_{s_k,t}^{-1}(\xi^{\Gamma}_{s_k}(x)-L_{s_k}(x)) =\\
    (\xi_{0}(x)-L_{0}(x))\prod_{i=0}^{k} e^{-(s_i^+-s_i)\gamma(\varphi_{ {s_i}}(x))}H_{s_i,s_i^+}^{-1},
  \end{multline}
  and $c^{\Gamma}$ is given by
  \begin{multline}\label{eq:discrete_consumption}
    c^{\Gamma}_{t}(x)\coloneqq
    \gamma(\varphi_{s_k}(x)) e^{-\gamma(\varphi_{s_k}(x))(t-s_k)}H_{s_k,t}^{-1}
    (\xi^{\Gamma}_{s_k}(x)-L_{s_k}(x))\\
    =\gamma(\varphi_{s_k}(x))\left(\xi^{\Gamma}_{t}(x)-
    L_{t}(x)\right),
  \end{multline}
  and no-arbitrage portfolio $\pi^{\Gamma}_{t}=(\pi^{\Gamma}_{t}(x))$ defined by:
  \begin{multline}\label{eq:discrete_portfolio}
    \pi^{\Gamma}_{t}(x)\coloneqq
    e^{-(t-s_k)\gamma(\varphi_{s_k}(x))}H_{s_k,t}^{-1} (\xi^{\Gamma}_{s_k}(x)-L_{s_k}(x))\kappa_{t}-\varpi_{t}(x)
    \\=\left(\xi^{\Gamma}_{t}(x)-
    L_{t}(x)\right)\kappa_{t}-\varpi_{t}(x),
  \end{multline}
  where $\varpi_{t}(x)$  is the portfolio that hedges the rate of  endowment  structure $Q$, and
  $L_t(x)$ is the current value of the future endowment structure.
\end{definition}

Passing to the limit as $\|\Gamma_n\|\to 0$ in \eqref{eq:discrete_wealth}--\eqref{eq:discrete_portfolio} yields the following theorem.

\begin{theorem}\label{thm:local_optimal_consumption_investment}
  Assume a financial market  $\mathfrak{M}=(P, \boldsymbol{\mu}, b,\sigma,\delta,\vartheta,r)$
  that satisfies the Smooth Market Condition (Condition~\ref{cond:smooth_market}).
  Assume an isoelastic consistent  preference structure for a
  population $U$ with effective impatience field  $\gamma(\varphi_t(x)(\omega))$
  of $U_t(\omega)(x)$, where  $\gamma(\cdot)$ is a function of
  type $C^{2}(\mathbb{D})$, satisfying
  Assumption~\ref{ass:population_preference_structure}.
  Let $\xi$ be the family of continuous semimartingales
  \begin{equation*}
    \xi_{t}(x)\coloneqq
    L_{t}(x)+ e^{-\int_{0}^t\gamma(\varphi_{u}(x))\,du}H_{t}^{-1}(\xi_{0}(x)-L_{0}(x))
  \end{equation*}
  and  $c$ the family of continuous semimartingales  given by
  \[
    c_{t}(x)\coloneqq
    \gamma(\varphi_{t}(x)) e^{-\int_0^t\gamma(\varphi_{u}(x))\,du}H_{t}^{-1}(\xi_{0}(x)-L_{0}(x)),
  \]
  and   no-arbitrage portfolio on stocks $\pi$ with generator $\pi_{t}(x)$ given by
  \begin{equation*}
    e^{-\int_{0}^t\gamma(\varphi_{u}(x))\,du}H_{t}^{-1}(\xi_{0}(x)-L_{0}(x))\kappa_{t}-\varpi_{t}(x)
  \end{equation*}
  for any  $x\in\mathbb{D}$, where   $\varpi_{t}(x)$  is the portfolio that hedges the rate of  endowment  structure $Q$.
  Then, $(\xi,c, Q)$ is a hedgeable rate of  consumption and rate of endowment
  structure, and $(\pi,c)\in\mathcal{A}(L,Q)$ is a no-arbitrage portfolio
  that arises as the pointwise limit of the sequence
  $(\xi^{\Gamma_n},c^{\Gamma_n},Q)$ associated with partitions $\Gamma_n$,
  where each $(\xi^{\Gamma_n},c^{\Gamma_n},Q)$ is optimal on each
  subinterval of $\Gamma_n$ in the sense of
  Theorem~\ref{thm:optimal_consumption_investment}
  (cf.\ Definition~\ref{def:structure_partition}),
  and $L_t(x)$ is the current value of future endowments for each type $x\in\mathbb{D}$.
\end{theorem}
\begin{proof}
  Fix $x\in\mathbb{D}$ and $T>0$. For any partition
  $\Gamma=\{0=s_0<\cdots<s_m=T\}$, define $\xi^\Gamma$, $c^\Gamma$, and $\pi^\Gamma$
  by rolling forward the optimal policy of Theorem~\ref{thm:optimal_consumption_investment}
  on each interval $[s_i,s_{i+1}]$ with frozen impatience
  $\gamma(\varphi_{s_i}(x))$ and initial wealth $\xi^\Gamma_{s_i}(x)$.
  This yields the explicit recursions \eqref{eq:discrete_wealth}--\eqref{eq:discrete_portfolio},
  and in particular $\big(\xi^\Gamma,c^\Gamma,Q\big)$ is hedgeable and
  $\big(\pi^\Gamma,c^\Gamma\big)$ is a state-tame (no-arbitrage) portfolio
  for every $\Gamma$.

  Using the cocycle properties $H_{s,t}=H_t/H_s$
  (and the flow/weight cocycle, cf.\ Appendix~\ref{appendix:aggregation}),
  the product form in \eqref{eq:discrete_wealth} can be written as
  \[
    \xi_t^\Gamma(x)-L_t(x)
    =H_t^{-1}\,(\xi_0(x)-L_0(x))\,
    \exp\!\Big(-\sum_{i:\,s_i<t}\gamma(\varphi_{s_i}(x))(s_{i+1}\wedge t-s_i)\Big).
  \]
  Since $u\mapsto \gamma(\varphi_u(x))$ is continuous, the Riemann sums converge as
  $\|\Gamma\|\to 0$ to $\int_0^t\gamma(\varphi_u(x))\,du$, hence
  $\xi_t^\Gamma(x)\to \xi_t(x)$ pointwise for all $t\in[0,T]$, with $\xi$ as stated in the theorem.
  Similarly, if $c^\Gamma$ and $\pi^\Gamma$ are the piecewise continuous processes defined by
  Equation~\eqref{eq:discrete_consumption} and Equation~\eqref{eq:discrete_portfolio}
  respectively, then  $c^\Gamma_t(x)\to c_t(x)$,
  and $\pi^\Gamma_t(x)\to \pi_t(x)$  as $\|\Gamma\|\to 0$ by continuity
  of $\gamma(\varphi_\cdot(x))$, $\kappa$, and $\varpi$.  By Assumption~\ref{ass:population_preference_structure}, it follows that $c_t(x)$ is a rate of consumption process, satisfying the integrability
  condition $\E\!\left[\int_{0}^{\infty}H_{t}\,c_{t}(x)\,dt\right]<\infty$
  for all $x\in \mathbb{D}$.

  Moreover, since the triple $(\xi^\Gamma,c^\Gamma,Q)$ satisfies the financing
  condition~\eqref{eqn:financiability}, and since Assumption~\ref{ass:population_preference_structure} holds, we may apply the stochastic dominated convergence theorem, and a localization argument yields that $(\xi,c,Q)$ satisfies~\eqref{eqn:financiability} and is therefore a
  wealth, consumption, and income structure.

  We also notice that $\pi_t(x)$ is a no-arbitrage portfolio since the original portfolio on each subinterval is a state-tame no-arbitrage portfolio \citep[see][Remark~6]{Londono2020a}.

  Finally, since $\xi_t^\Gamma(x)\geq L_t(x)$ on each subinterval, it follows
  that $\xi_t(x)\geq L_t(x)$, implying that $(\pi,c)$ is admissible for
  $(L,Q)$.
\end{proof}

\begin{corollary}[Local optimal aggregate]\label{cor:optimal_local_aggregate}
  Assume the conditions of Theorem~\ref{thm:local_optimal_consumption_investment} with optimal
  solution $(\xi,c, Q)$, where Assumption~\ref{ass:population_preference_structure} holds.
  Consider the aggregate process $\eta_t=\eta^\mu_t$ associated with the field defined by
  \begin{equation}\label{eq:aggregate_process_optimal}
    \eta_t
    \coloneqq  \int_{\mathbb{D}}
    e^{-\int_{0}^t\gamma(\varphi_{u}(x))\,du}
    y(x)\,
    d\mu(x), \qquad t\ge0,
  \end{equation}
  where $\gamma(x)$ is the effective impatience rate from the
  preference structure, and $y(x)\coloneqq \xi_{0}(x)-L_{0}(x)$ is the initial total wealth of type $x$. Then the population-weighted aggregate of the optimal wealth processes satisfies
  \begin{equation}\label{eq:aggregate_const_optimal}
    \xi^\mu_{t}-L^\mu_t
    \coloneqq  \int_{\mathbb{D}}
    \big(\xi_t(x)-L_t(x)\big)\,d\mu(x)
    = \frac{1}{H_{t}}\eta_{t},
  \end{equation}
  where
  \[
    L^\mu_t
    \coloneqq  \int_{\mathbb{D}} L_{t}(x)\,d\mu(x)
  \]
  is the current value of future endowments $L_t(x)$.

  The population-weighted aggregate of the optimal consumption process is
  given in terms of $\eta$ by
  \begin{equation}\label{eq:aggregatecons_const_optimal}
    c^\mu_{t}
    \coloneqq  \int_{\mathbb{D}} c_t(x)\,d\mu(x)
    = \frac{-\partial_t\eta_t}{H_{t}},
  \end{equation}
  where
  \[
    -\partial_t\eta_t =
    \int_{\mathbb{D}}
    \big(\gamma(\varphi_t(x))\big)\,
    \exp\!\Big(-\int_0^t \gamma\big(\varphi_u(x)\big)\,du\Big)\,
    y(x)\,d\mu(x),\qquad t\ge0.
  \]

  Moreover, the population-weighted aggregate optimal portfolio in stocks is
  \[
    \pi_{t}^\mu
    \coloneqq  \int_{\mathbb{D}} \pi_t(x)\,d\mu(x)
    = H_{t}^{-1}\eta_{t}\,\kappa_{t}
    - \varpi^\mu_{t},\qquad
    \varpi^\mu_{t}
    \coloneqq  \int_{\mathbb{D}} \varpi_t(x)\,d\mu(x)
  \]
  where  $\varpi^\mu_{t}$ is the aggregate hedging portfolio associated with $\varpi_{t}(x)$ that hedges the rate of  endowment  structure $Q$.
\end{corollary}

Next, we study the concept of equilibrium in our framework.

\begin{definition}\label{def:equilibrium}
  Assume an  economy  $\mathcal{E}$ with underlying market $\mathfrak{M}$ and population structure $\boldsymbol{\mu}$
  (with initial distribution $\mu$) as in Definition~\ref{def:financial_market},
  and \emph{aggregate  hedgeable (by state-tame portfolios)
    rate of consumption and rate of endowment structure
    $(\xi,c, Q)$}   and current value of future endowments $L$ with no-arbitrage portfolio
  $\pi$.   We say that $\mathcal{E}$   is an \emph{equilibrium} if   the
  following are satisfied for  $\mu$:
  \begin{description}
    \item[Clearing of the money market]
          \begin{equation*}
            \xi^\mu_{t}-\pi^\mu_{t}=0
          \end{equation*}
          for all $t\geq 0$, a.s.
    \item[Clearing of the commodity market]
          \[
            c^\mu_{t}=
            Q^\mu_{t}+D^\mu_t>0
          \]
          for all $0\leq t$ a.s., where $D^\mu_t=\delta^\mu_{t}P^\mu_{t}$ is the aggregate rate of dividend income.
    \item[Clearing of the stock market]
          \[
            \pi^\mu_{t}=P^\mu_{t}
          \]
          a.s.,
  \end{description}
  where $\xi^\mu_{t}$, $\pi^\mu_{t}$, $c^\mu_{t}$, $Q^\mu_{t}$, and $L^\mu_{t}$ are the
  aggregate processes associated with $\xi_{t}(x)$, $\pi_{t}(x)$, $c_{t}(x)$, $Q_{t}(x)$ and $L_t(x)$
  respectively as in Definition~\ref{def:aggregation}.
\end{definition}

\begin{definition}\label{def:short_equilibrium}
  Assume an economy  $\mathcal{E}$ with underlying market $\mathfrak{M}$
  as defined by Theorem~\ref{thm:local_optimal_consumption_investment},
  a population structure $\boldsymbol{\mu}$ (with initial distribution $\mu$)
  and \emph{hedgeable (by state-tame portfolios) rate of consumption and rate of endowment
    structures $(\xi,c, Q)$}, current value of future endowments $L$ and no-arbitrage portfolio
  $\pi$
  \emph{defined by Theorem~\ref{thm:local_optimal_consumption_investment}}.
  We say that $\mathcal{E}$  is a \emph{short-horizon Duesenberry equilibrium}
  if the equations  and conditions defining the clearing of the money market, the clearing
  of the commodity market and the clearing of the stock market of
  Definition~\ref{def:equilibrium} hold.
\end{definition}

\begin{theorem}\label{thm:local_equilibrium}
  Let $\varphi$ be a $\mathcal{C}(\mathbb{D};\mathbb{D})$-valued Brownian
  flow of $C^2(\mathbb{D})$ diffeomorphisms, and let $\boldsymbol{\mu}$ be a population
  structure on $\mathbb{D}$, with initial population measure $\mu$ (see Definition~\ref{def:aggregation}).
  Assume a dynamic, consistent  isoelastic preference structure for a population $U_{s,t}(x)$
  as in Definition~\ref{def:utility}, with corresponding effective impatience
  rate  $\gamma(x)$  of type $C^{2}(\mathbb{D})$
  satisfying Assumption~\ref{ass:population_preference_structure}.
  \begin{enumerate}[label=(\alph*)]
    \item \textbf{(Characterization.)} Suppose $\mathcal{E}$ is a short-horizon Duesenberry equilibrium economy
          with underlying market $\mathfrak{M}$ with population structure $\boldsymbol{\mu}$
          that satisfies the smooth market condition
          (Condition~\ref{cond:smooth_market}) with   \emph{aggregate hedgeable
            (by state-tame portfolios) rate of consumption and rate of endowment
            structures $(\xi^\mu,c^\mu, Q^\mu)$}  and no-arbitrage portfolio  $\pi^\mu$
          \emph{defined by Theorem~\ref{thm:local_optimal_consumption_investment}}.
          Define the
          population-weighted effective aggregate process $\eta_{t}=\eta^\mu_{t}$
          \begin{equation}\label{eq:wealth_distribution_equilibrium}
            \eta_{t}
            \coloneqq  \int_{\mathbb{D}}
            \exp\!\Big(-\int_{0}^t \gamma(\varphi_{u}(x))\,du\Big)
            \,y(x)\,\,d\mu(x),
          \end{equation}
          and  the aggregate Duesenberry loading
          \begin{equation}\label{eq:varkappa_equilibrium}
            -\partial_t\eta_t =
            \int_{\mathbb{D}}
            \gamma(\varphi_t(x))\,
            \exp\!\Big(-\int_0^t \gamma\big(\varphi_u(x)\big)\,du\Big)\,
            y(x)\,d\mu(x),\qquad t\ge0.
          \end{equation}
          Then, there exists a $C^2$ function $y(x)$ with
          \begin{equation}\label{eq:kernel2}
            \int_{\mathbb{D}} y(x)\,\gamma(x)\,d\mu(x)=I^\mu_0=Q^\mu_{0}+D^{\mu}_0>0,
          \end{equation}
          where $D^{\mu}_t=\delta^\mu_{t}P^\mu_{t}$, and $I_t^\mu=Q_t^\mu+D_t^\mu$ is the
          aggregate income process.  Moreover, the  state-price process satisfies
          \begin{equation}\label{eq:state_price2}
            H^\mu_t=\frac{-\partial_t\eta_t}{I^\mu_t},
          \end{equation}
          and
          \begin{equation}\label{eq:aggregate_wealth2}
            P^\mu_{t}
            =
            \frac{I^\mu_t}{\partial_t\eta_t}\int_t^{\infty}\Psi^{\mu_t}_{s}\,ds
            -\frac{I^\mu_t}{\partial_t\eta_t}\,\eta_t=
            \frac{I^\mu_t}{\partial_t\eta_t}\left(\int_t^{\infty}\Psi^{\mu_t}_{s}\,ds-\eta_t\right),
          \end{equation}
          where
          \[
            Q^\mu_{t} = \int_{\mathbb{D}} Q^\mu_t(x)\,d\mu(x),
          \]
          and  $\Psi^\mu_t$ is
          defined as
          \begin{equation}
            \Psi^\mu_{t}
            =
            -\E\!\left[\frac{Q^\mu_{t}}{I^\mu_{t}}\,\partial_t\eta_t\right]
            =
            \int_{\mathbb{D}}\Psi^\mu_{t}(x)\, d\mu(x),
          \end{equation}
          with
          \begin{equation*}
            \Psi^\mu_{t}(x)
            =
            -\E\!\left[\frac{Q^\mu_t(x)}{I^\mu_t}\,\partial_t\eta_t\right]
            =
            \E\!\left[H^\mu_t\,Q_t^\mu(x)\right].
          \end{equation*}

    \item \textbf{(Existence.)} Assume non-negative (semimartingale)
          random fields $Q_t^{\mu}(x)$ and $I_t^\mu(x)$
          of type $C^{2}(\mathbb{D})$ consistent
          with the population structure $\boldsymbol{\mu}$, satisfying
          (either) Assumption~\ref{ass:discrete_app} or Assumption~\ref{ass:SF_app}.
          Assume
          $0\leq Q_t^\mu(x)< I_t^\mu(x)$  with $I_t^\mu(x)>0$ for all
          $x\in\mathbb{D}$, and $t\geq 0$.
          Moreover, assume that there exists a $C^2$ function  $y(x)>0$ for all $x$ with
          \begin{equation}\label{eq:budget_existence}
            0<\int_{\mathbb{D}}y(x)\gamma(x)\,d\mu(x)=\int_{\mathbb{D}}I_0^\mu(x)\,d\mu(x) <\infty.
          \end{equation}
          Furthermore,
          we define $\Psi^\mu_{t}(x)$ and $\Psi^\mu_{t}$ by the following equations:
          \begin{equation}\label{eqn:Psi_x}
            \Psi^\mu_{t}(x)
            =
            -\E\!\left[\frac{Q^\mu_t(x)}{I_t^\mu}\,\partial_t\eta_t\right], \qquad
            \Psi^\mu_{t}=
            \int_{\mathbb{D}}\Psi^\mu_{t}(x)\, d\mu(x)
          \end{equation}
          where
          \begin{equation}\label{eq:uniform_integrability}
            \int_0^{\infty}\int_{\mathbb{D}}\Psi^\mu_{t}(x)\,d\mu(x)\,dt<\infty\qquad\text{and} \qquad \int_0^{\infty}\Psi^\mu_{t}(x)\,dt<\infty,\qquad \forall x\in\mathbb{D}.
          \end{equation}
          We define   the market $\mathfrak{M} = (P^\mu, \boldsymbol{\mu}, b^\mu, \sigma^\mu, \delta^\mu, \vartheta^\mu, r^\mu)$   by the following equations:
          \begin{equation}\label{eq:state_price2_existence}
            H^\mu_t=\frac{-\partial_t\eta_t}{I_t^\mu},
          \end{equation}
          and

          \begin{equation}\label{eq:aggregate_wealth2_existence}
            P^\mu_{t}=\frac{I^\mu_t}{\partial_t\eta_t}\left(\int_t^{\infty}\Psi^{\mu_t}_{s}\,ds-\eta_t\right), \qquad \delta_t^{\mu}P^\mu_t=D_t^\mu=I_t^\mu-Q_t^\mu,
          \end{equation}
          where $I^\mu_{t}$ and $Q^\mu_{t}$ are the aggregate processes of $I_t^\mu(x)$, and $Q_t^\mu(x)$
          respectively, and it is assumed that there exists a \emph{continuous}
          semimartingale process $\kappa^{\mu}_{t}$, consistent with the population structure $\boldsymbol{\mu}$, with
          \begin{equation}\label{eq:smooth_market2}
            \kappa^\mu_{t}\sigma^\mu_{t}=\vartheta^\mu_{t}.
          \end{equation}

          Then, there exists a short-horizon Duesenberry equilibrium  $\mathcal{E}$
          with underlying market $\mathfrak{M}$ with population structure $\boldsymbol{\mu}$ and
          a  dynamic, consistent isoelastic preference structure for a
          population with effective impatience rate $\gamma(x)$ with type space $\mathbb{D}$.
    \item In either case, the underlying market is free of state-tame arbitrage.
  \end{enumerate}
\end{theorem}

\begin{proof}[Proof of Theorem~\ref{thm:local_equilibrium}]
  We divide the argument into two parts.
  First, we prove the \emph{characterization}: starting from a short-horizon
  Duesenberry equilibrium, we obtain the identities
  \eqref{eq:wealth_distribution_equilibrium}--\eqref{eq:kernel2} and the pricing formulas \eqref{eq:state_price2}--\eqref{eq:aggregate_wealth2}.
  Then, we prove the \emph{existence} direction by reversing the construction. Finally, we prove
  no-arbitrage of the market in both cases.

  \medskip
  \noindent\textbf{Characterization.}
  Assume a short-horizon Duesenberry equilibrium $\mathcal{E}$  with underlying market
  $\mathfrak{M}$  that satisfies the Smooth Market Condition
  (Condition~\ref{cond:smooth_market}).
  Under the hypotheses of Theorem~\ref{thm:local_optimal_consumption_investment},
  for each agent of (initial) type $x\in\mathbb{D}$, the locally optimal wealth, consumption, and portfolio processes are given by
  $(\xi^\mu,c^\mu, Q^\mu)$ and $\pi^\mu$ as follows:
  \begin{equation}\label{eq:individual_wealth}
    \xi^\mu_t(x)=L^\mu_t(x)+ e^{-\int_{0}^t\gamma(\varphi_{u}(x))\,du}(H^\mu_{t})^{-1}y(x)
  \end{equation}
  where  the \emph{initial net wealth} $y(x)=\xi^\mu_0(x)-L^\mu_0(x)$ is a function of class $C^2$ and
  \begin{equation}\label{eqn:individual_consumtion}
    c^\mu_{t}(x)=
    \gamma(\varphi_{t}(x)) e^{-\int_0^t\gamma(\varphi_{u}(x))\,du}(H^\mu_{t})^{-1}y(x)
  \end{equation}
  and a state-tame (no-arbitrage) portfolio in the sense of Theorem~\ref{thm:local_optimal_consumption_investment}
  and Remark~\ref{rem:state_tameness} given by:
  \begin{equation}\label{eq:indiv_port_opt}
    \pi^\mu_t(x)
    = e^{-\int_0^t \gamma(\varphi_u(x))\,du}\,
    (H^\mu_{t})^{-1}y(x)\,\kappa^\mu_t
    \;-\; \varpi^\mu_t(x).
  \end{equation}

  Now, we aggregate over the (random) population at time $t$ the  second term of the right-hand side of Equation~\eqref{eq:individual_wealth}. This yields Equation~\eqref{eq:wealth_distribution_equilibrium}
  by Corollary~\ref{cor:optimal_local_aggregate}.  Similarly, integrating equations \eqref{eq:individual_wealth} and \eqref{eqn:individual_consumtion} over the (time-$t$) distribution $\mu_t$ gives the \emph{aggregate} processes
  \begin{equation}\label{eq:agg_wealth}
    \xi_t^\mu
    =L_t^\mu \;+\; \big(H_t^\mu\big)^{-1}\eta_t
  \end{equation}
  and
  \begin{equation}\label{eq:agg_cons_tmp}
    -\partial_t\eta_t =
    \int_{\mathbb{D}}
    \gamma\big(\varphi_t(x)\big)\,
    \exp\!\Big(-\int_0^t \gamma\big(\varphi_u(x)\big)\,du\Big)\,
    y(x)\,d\mu(x).
  \end{equation}

  Note that $t\mapsto \exp(-\int_0^t \gamma(\varphi_{u}(x))du)$
  is non-increasing for each $x$, hence, so is $\eta_t$.  Also,
  by Corollary~\ref{cor:aggregate_semimartingale_app} $\eta_t$, and $-\partial_t \eta_t$
  are consistent semimartingales with respect to the population structure
  $\boldsymbol{\mu}$. Moreover, by the dominated convergence theorem
  $\eta_t\to 0$ a.s.

  Aggregating \eqref{eqn:individual_consumtion} against the population structure gives
  \begin{equation}\label{eq:agg_cons_deta}
    c_t^\mu
    =
    -\big(H^\mu_t\big)^{-1}\,\partial_t\eta_t.
  \end{equation}
  By the clearing of the commodity market $c_t^\mu=I_t^{\mu}>0$ for all $t\ge0$ (a.s.), we obtain
  \begin{equation}\label{eq:state_price_proof}
    H^\mu_t
    =
    \frac{-\partial_t \eta_t}{I^\mu_t},
  \end{equation}
  which is Equation~\eqref{eq:state_price2},  where $Q_t^\mu$ is the aggregate (hedgeable) endowment rate and
  $D^\mu_t=\delta^\mu_t P^\mu_t$ is the aggregate rate of dividend income. Using the clearing of the commodity
  market at $t=0$ and using~\eqref{eq:agg_cons_tmp} we obtain Equation~\eqref{eq:kernel2}.

  Note that Equation~\eqref{eq:state_price_proof}, the dominated convergence
  theorem, Assumption~\ref{ass:population_preference_structure}, and Equation~\eqref{eq:kernel2} imply that
  \begin{equation}
    \label{eq:HD_bound2}
    0\le \int_0^\infty H^{\mu}_tQ_t^{\mu}\, dt=\int_0^{\infty} \frac{Q^{\mu}_t}{I^{\mu}_t}(-\partial_t\eta_t)\,dt\le\int_0^{\infty}(-\partial_t\eta_t)\,dt=\eta_0<\infty
  \end{equation}
  By Equation~\eqref{eq:HD_bound2}, since $Q^{\boldsymbol{\mu}}$ is a hedgeable rate of endowment structure with current value of future endowments $L^{\boldsymbol{\mu}}$
  (with \emph{no-arbitrage} portfolio $\pi^{\boldsymbol{\mu}}_Q$), using Fubini's theorem and the cocycle property, it follows that
  \begin{equation}\label{eq:current_value_labor}
    L_t^\mu = \int_{\mathbb{D}} L_t(x)\,d\mu(x)
    =
    -\,(H^{\mu}_t)^{-1}\,
    \E\!\left[
      \int_t^\infty H^{\mu}_u \,Q_u^\mu\,du
      \;\Big|\;\mathcal{F}_t
      \right]=-\,(H^\mu_t)^{-1}\,\int_t^\infty \Psi^{\mu_t}_{u}\, du
  \end{equation}
  Equation~\eqref{eq:aggregate_wealth2} is a direct consequence of Corollary~\ref{cor:optimal_local_aggregate},
  and Equation~\eqref{eq:current_value_labor}.

  This completes the derivation of \eqref{eq:aggregate_wealth2} from the market-clearing conditions together with \eqref{eq:state_price2}.

  On the other hand, by It\^o's Lemma, Equation~\eqref{eq:state_price2},
  Equation~\eqref{eq:aggregate_wealth2}, and since $\eta_t\to 0$ it follows that
  \begin{multline*}
    E\Big[\int_0^{\infty}\big(1-\frac{Q^\mu_u}{I_u^\mu}\big)(-\partial_u\eta_u)\,d u\mid\mathcal{F}_t\Big]=
    H_t^\mu P_t^\mu+\int_0^tH_u^\mu D_u^\mu\,du=\\
    P_0^\mu+\int_0^t H_u^\mu P_u^\mu\Big(b^\mu_u+\delta_u^\mu-r^\mu_u-(\sigma_u^\mu)^\intercal\vartheta_u^\mu\Big)\,du+
    \int_0^tH^\mu_uP_u^\mu\left((\sigma^\mu_u)^\intercal-(\vartheta^\mu_u)^\intercal\right)dW(u)
  \end{multline*}
  It follows that $b^\mu_u+\delta_u^\mu-r^\mu_u-(\sigma_u^\mu)^\intercal\vartheta_u^\mu=0$, proving
  that there are no state-tame arbitrage opportunities.
  \medskip

  \noindent\textbf{Existence.}
  Assume non-negative random fields $Q^\mu(x)$ and $I_t^\mu(x)$ of type $C^{2}(\mathbb{D})$,
  with the property that there exist  a $C^2$ function  $y(x)>0$ for all $x$
  satisfying Equation~\eqref{eq:kernel2}.
  Define  $\Psi^\mu_{t}(x)$, and $\Psi^\mu_{t}$ by Equation~\eqref{eqn:Psi_x}, assuming  Equation~\eqref{eq:uniform_integrability}.
  Define,  the effective aggregate process $\eta_t>0$ by Equation~\eqref{eq:wealth_distribution_equilibrium}, and define
  $H_t^{\mu}$  by Equation~\eqref{eq:state_price2_existence}.
  By \eqref{eq:kernel2} and \eqref{eq:wealth_distribution_equilibrium},
  \[
    \int_{\mathbb{D}}\gamma(x)\,y(x)\,d\mu(x)
    =I_0^\mu,
  \]
  hence $H^\mu_0=1$ by Equation~\eqref{eq:state_price2}. Since $\gamma(x)>0$, $y(x)>0$
  and $I^\mu_t(x)>0$ for all $x$, $H^\mu_t>0$.

  Let $(r^\mu,\vartheta^\mu)$ be the coefficients obtained from the It\^{o} decompositions of $H^{\mu}$.
  Also,
  \begin{multline*}
    \int_t^{\infty}\Psi^{\mu}_{s}\,ds=
    E\big[\int_0^{\infty}\int_{\mathbb{D}}\frac{Q^\mu_u(x)}{I^\mu_u}(-\partial_u \eta_u)\,d\mu(x)\,du\Big]<\\
    E\Big[\int_0^\infty\Big(-\partial_u\eta_u\Big)\,du\Big]=\E[\eta_0]<\infty
  \end{multline*}
  Moreover,
  \begin{equation}\label{eq:auxiliar_price}
    \int_t^{\infty}\Psi^{\mu_t}_{s}\,ds-\eta_t=E\Big[\int_t^{\infty}\frac{Q_u^{\mu}}{I^\mu_u}(-\partial_u \eta_u)\,du-\int_t^{\infty}(-\partial_u\eta_u)\,du\mid\mathcal{F}_t\Big]<0
  \end{equation}
  It follows from Equation~\eqref{eq:aggregate_wealth2_existence}, and
  Equation~\eqref{eq:auxiliar_price}:
  \begin{equation*}
    P^\mu_{t}=\frac{I^\mu_t}{\partial_t\eta_t}\left(\int_t^{\infty}\Psi^{\mu_t}_{s}\,ds-\eta_t\right)>0.
  \end{equation*}

  Define $(b^\mu,\sigma^\mu)$ the coefficients obtained satisfying Assumption~\ref{ass:discrete_app} or~\ref{ass:SF_app} according to the
  underlying hypothesis of the theorem, and define $\delta_t^\mu\coloneqq D^\mu_t/P^\mu_t$ as the aggregate dividend yield, where $D_t^\mu=I_t^\mu-Q_t^\mu>0$ is the aggregate dividend income.  Assume that~\eqref{eq:smooth_market2} holds
  with a consistent (relative to the population structure $\boldsymbol{\mu}$) process
  $\kappa^{\mu}_{t}$ such that $\kappa^{\mu}_{t}\sigma^{\mu}_{t}=\vartheta^{\mu}_{t}$. In
  particular the market defined by $P_t^\mu$ and $H_t^\mu$ above satisfies
  the smooth market condition, and therefore Theorem~\ref{thm:local_optimal_consumption_investment} and
  Corollary~\ref{cor:optimal_local_aggregate} apply.
  Consider the consumer problem in the market $\mathfrak{M}=(P^\mu, \boldsymbol{\mu}, b^\mu,\sigma^\mu,\delta^\mu,\vartheta^\mu,r^\mu)$
  as above, with the given consistent isoelastic preference structure. Define
  $\big(\xi^\mu(x),c^\mu(x),\pi^\mu(x)\big)$
  the formulas of Theorem~\ref{thm:local_optimal_consumption_investment}
  satisfying the explicit formulas \eqref{eq:individual_wealth}--\eqref{eqn:individual_consumtion}
  (and the corresponding portfolio representation), with initial net wealth
  $y^\mu(x)=y(x)$, and initial financial wealth $\xi^\mu_0(x)=y(x)+L_0^\mu(x)$,
  and labor income $Q_t^\mu(x)$, so
  \[-L^\mu_t(x)=\frac{I^\mu_t}{\partial_t\eta_t}\int_t^{\infty}\Psi^{\mu_t}_{s}(x)\,ds=(H_t^{\mu})^{-1}\E\!\left[
      \int_t^\infty H^{\mu}_u \,Q_u^\mu(x)\,du
      \;\Big|\;\mathcal{F}_t
      \right],
  \]
  is the current value of future labor at time $t$ in the constructed market.
  It follows by Corollary~\ref{cor:optimal_local_aggregate}, that $\xi_t^\mu=L_t^\mu+(H_t^\mu)^{-1}\eta_t$,
  and therefore by Equation~\eqref{eq:aggregate_wealth2_existence}, $P_t^\mu=\xi_t^\mu$.
  Also, by It\^o's Lemma, and aggregation, it follows from Equation~\eqref{eq:aggregatecons_const_optimal}

  \begin{multline}\label{eq:labor_semimartingale_existence}
    H_t^\mu L_t^\mu-\int_0^tH^{\mu}_u Q_u^\mu \,du=\\
    -E\Big[\int_0^{\infty}H^\mu_uQ^\mu_u\,du\Big]+
    \int_0^tH_u^\mu\Big( (\pi_Q^\mu)_u(\sigma_u^\mu)^{\intercal}-L_u^\mu(\vartheta^\mu_u)^\intercal\Big) dW(u),
  \end{multline}
  where $\pi_Q^\mu$ is the aggregate  portfolio (that hedges $L^\mu_t$) implied by the representation as local
  martingale of the term on the left of the last equation.

  We observe that
  \begin{multline*}
    H_t^\mu P_t^\mu+\int_0^tH_u^\mu D_u^\mu\,du=\\
    E\Big[\Big(\int_t^{\infty}(-\partial_u\eta_u)\,du-\int_t^{\infty}\frac{Q_u^{\mu}}{I^\mu_u}(-\partial_u \eta_u)\,du\Big)\mid\mathcal{F}_t\Big]+\int_0^t \Big(I_u^\mu-Q_u^\mu\Big)\frac{(-\partial_u \eta_u)}{I^\mu_u}\,du=\\
    E\Big[\Big(\int_0^{\infty}(-\partial_u\eta_u)\,du-\int_0^{\infty}\frac{Q_u^{\mu}}{I^\mu_u}(-\partial_u \eta_u)\,du\Big)\mid\mathcal{F}_t\Big]
  \end{multline*}
  is a martingale.  On the other hand, It\^o's Lemma implies
  \begin{multline}\label{eq:current_value_price}
    H_t^\mu P_t^\mu+\int_0^tH_u^\mu D_u^\mu\,du=P_0^\mu+\int_0^t H_u^\mu P_u^\mu\Big(b^\mu_u+\delta_u^\mu-r^\mu_u-(\sigma_u^\mu)^\intercal\vartheta_u^\mu\Big)\,du+\\
    \int_0^tH^\mu_uP_u^\mu\left((\sigma^\mu_u)^\intercal-(\vartheta^\mu_u)^\intercal\right)dW(u)
  \end{multline}
  Since the left-hand side of Equation~\eqref{eq:current_value_price} is a martingale,
  it follows that  the bounded variation part is $0$, implying that
  \[b_u^\mu+\delta_u^\mu-r_u^\mu=(\sigma_u^\mu)^\intercal\vartheta_u^\mu
  \]
  and then the market, as defined above, is free of state-tame arbitrage opportunities.
  From Theorem~\ref{thm:local_optimal_consumption_investment}, and Corollary~\ref{cor:optimal_local_aggregate}
  and Equation~\eqref{eq:state_price2_existence}, it follows that
  \begin{equation}\label{eq:consumption_existence}
    c^\mu_{t} = \frac{-\partial_t\eta_t}{H^\mu_{t}}=I_t^\mu
  \end{equation}
  Moreover from Theorem~\ref{thm:local_optimal_consumption_investment}, and Corollary~\ref{cor:optimal_local_aggregate},
  and Equation~\eqref{eq:consumption_existence} it follows that
  \begin{equation}\label{eq:optimal_wealth_existence}
    H^\mu_t\xi_t^\mu+\int_0^tH_u^\mu (c_u^\mu-Q_u^\mu)\,du =\xi^\mu_0+\int_0^tH_u^\mu\Big(\pi_u^\mu(\sigma_u^\mu)^\intercal-\xi_u^\mu(\vartheta_u^\mu)^\intercal\Big)\, dW(u)
  \end{equation}
  Taking differences of equations~\eqref{eq:optimal_wealth_existence} and~\eqref{eq:current_value_price},
  and since $\xi_t^\mu=P_t^\mu$, and since $D^\mu_t=I_t^\mu-Q^\mu_t=c^\mu_t-Q^\mu_t$ by Equation~\eqref{eq:consumption_existence} it follows that $P^\mu_t=\pi_t^\mu$, proving the clearing of the money market,
  and the clearing of the stock market.
  This completes the existence part and therefore the proof.
\end{proof}

\begin{corollary}\label{cor:properties_local_equilibrium}
  Assume the conditions of Theorem~\ref{thm:local_equilibrium}. Let
  \begin{equation}\label{eq:wealth_distr}
    \eta_{t}
    \coloneqq \int_{\mathbb{D}}
    \exp\!\left(-\int_{0}^t\gamma(\varphi_{u}(x))\,du\right)\,y(x)\,d\mu(x)
  \end{equation}
  be the population-weighted effective aggregate process.  Then the state-price process and the consumption--wealth ratio satisfy
  \begin{equation}\label{eq:price_discount}
    H_{t}=\frac{-\partial_t\eta_t}{c_{t}},\qquad
    \frac{c_t}{P^W_t}=\frac{-\partial_t\eta_t}{\eta_t}=-\partial_t\log(\eta_t),
  \end{equation}
  where $c_t=c_t^\mu$ is aggregate consumption and
  \[
    P_t^L
    \coloneqq  (H_t)^{-1}\E\!\left[\int_t^{\infty}H_u\,Q_u^\mu\,du\;\middle|\;\mathcal F_t\right]
    = -\,L_t^\mu,
    \qquad
    P_t^W\coloneqq P_t+P_t^L=(H_t)^{-1}\eta_t
  \]
  is the total wealth (physical wealth plus labor wealth) owned by society at time $t$.
  Moreover, since $\eta$ is a.s.\ absolutely continuous  in $t$,
  the volatility $\sigma_t^W$ of $P_t^W$  is equal to the market price of risk $\vartheta_t^\mu$.
\end{corollary}

\begin{remark}\label{rem:consumption_wealth_ratio}
  The right-hand side of Equation~\eqref{eq:price_discount} implies
  that the consumption--wealth ratio is the \emph{instantaneous proportional decay rate} of $\eta_t$,
  where $\eta_t$ represents the impatience-adjusted social wealth:
  it is a cross-sectional average of initial wealth $y(x)$, weighted by the cumulative
  impatience factor along each trajectory through the effective impatience field $\gamma(\varphi_u(x))$.
  A consequence of the latter is that the only factor that accounts for changes in the consumption--wealth ratio is derived
  from changes in the preference structure of the population,
  and not from factors such as productivity, population growth, or inflation, at least directly.
  See~\citet{LSV2013} for empirical analysis of the
  wealth--consumption ratio, and~\citet{LettauLudvigson2001cay} for the closely related
  \emph{cay} variable, which has been shown to forecast stock returns and whose
  fluctuations, in the present framework, are driven entirely by shifts in the
  cross-sectional impatience structure~$\eta_t$.

  \remarkparagraph{Volatility decomposition.}
  If $\sigma^c_t$ denotes the volatility of aggregate consumption
  (relative diffusion) and
  \[
    \sigma^{-\partial \eta}_t=\frac{1}{-\partial_t\eta_t}\int_{\mathbb{D}}
    \big(\nabla \gamma(\varphi_t(x))\varrho(\varphi_{t}(x))\big)\,
    \exp\!\Big(-\int_0^t \gamma\big(\varphi_u(x)\big)\,du\Big)\,
    y(x)\,d\mu(x)
  \]
  is the volatility (relative diffusion) of $-\partial_t\eta_t$, then the volatility $\sigma_t^W$ of total
  wealth $P^W_t$ (relative diffusion) satisfies
  \[
    \sigma_t^W=\vartheta_t=\sigma_t^c-\sigma_t^{-\partial\eta},
  \]
  where $\vartheta_t$ is the market price of risk process.  Hence, asset-price risk premia and market volatility (which are governed by the diffusion of the
  state-price process) need not be tied to consumption volatility: even if aggregate consumption is
  smooth (small $\|\sigma_t^{c}\|$), aggregate wealth (or the state price) can be highly volatile whenever impatience risk generates
  large fluctuations in $-\partial_t\eta_t$.
  This provides a direct valuation channel through which time-varying effective impatience can reconcile
  low consumption volatility with volatile markets.
  We note that this mechanism is distinct from the long-run risk channel of~\citet{BansalYaron2004LongRunRisk}, where persistent shocks to consumption growth
  drive asset-price volatility, and from the external habit mechanism
  of~\citet{CampbellCochrane1999ForceHabit}, where a time-varying surplus ratio
  amplifies risk premia.  In our framework, the additional source of market volatility arises from heterogeneous and time-varying impatience across
  the population, encoded in the term $\sigma_t^{-\partial\eta}$.
\end{remark}

\begin{remark}\label{rem:two_type_example}
  \remarkparagraph{Heterogeneity is not enough.}
  Let $\mathbb{D}=\mathbb{D}_1\sqcup\mathbb{D}_2$ with $\mathbb{D}_1,\mathbb{D}_2$
  disjoint open sets, let $\gamma\equiv\gamma_i$ on $\mathbb{D}_i$  and write $y_i\coloneqq \int_{\mathbb{D}_i}y(x)\,d\mu(x)$. Since
  $u\mapsto\varphi_u(x)$ is continuous, each trajectory stays in its connected
  component, so $\gamma(\varphi_u(x))=\gamma(x)$ and
  \[
    \eta_t=y_1e^{-\gamma_1t}+y_2e^{-\gamma_2t},
    \qquad
    -\partial_t\eta_t=y_1\gamma_1e^{-\gamma_1t}+y_2\gamma_2e^{-\gamma_2t}
  \]
  are deterministic: the relative-income representation of
  Remark~\ref{rem:relation_joneses} is available, yet
  $\sigma^{-\partial\eta}_t=0$ and the impatience-risk premium of
  Remark~\ref{rem:equity_premium} vanishes, whether the
  household solves a fixed-horizon problem or rolls short-horizon problems
  forward. What activates the premium is not heterogeneity of impatience as
  such, but its variation along the flow: for non-constant
  $\gamma\in C^{2}(\mathbb{D})$ on a connected $\mathbb{D}$, transport moves each
  type through regions of differing impatience and $-\partial_t\eta_t$ acquires
  the diffusion generated by $\nabla\gamma$.
\end{remark}

\begin{remark}\label{rem:equity_premium}
  \remarkparagraph{Equity premium decomposition.}
  The volatility decomposition $\vartheta_t=\sigma_t^c-\sigma_t^{-\partial\eta}$
  of Remark~\ref{rem:consumption_wealth_ratio}, together with the
  no-arbitrage relation~\eqref{E:wviability}, yields an explicit
  decomposition of the equity premium of any traded asset (or portfolio)
  with price process~$P^\Sigma_t$, rate of return~$b_t^\Sigma$,
  dividend yield~$\delta_t^\Sigma$, and volatility vector~$\sigma_t^\Sigma$:
  \begin{equation}\label{eq:equity_premium_decomposition}
    b_t^\Sigma+\delta_t^\Sigma-r_t
    \;=\;
    \underbrace{(\sigma_t^c)^\intercal\,\sigma_t^\Sigma\,}_
    {\displaystyle\text{consumption risk premium}}
    \;-\;
    \underbrace{(\sigma_t^{-\partial\eta})^\intercal\,\sigma_t^\Sigma\,}_
    {\displaystyle\text{impatience risk premium}}
  \end{equation}
  The first term is the classical Consumption-CAPM premium
  \citep[see][]{Breeden79}: the covariance of aggregate
  consumption growth with asset returns.  The second term is the
  \emph{impatience risk premium}: the covariance of the growth of the
  Duesenberry loading~$-\partial_t\eta_t$ with asset returns.

  \remarkparagraph{From market volatility to aggregate wealth volatility.}
  A convenient way to read~\eqref{eq:equity_premium_decomposition} is that equity premia are governed by an \emph{inner product of volatilities}.
  In fact, as a consequence of Equation~\eqref{eq:price_discount}, the market price of risk
  $\vartheta_t=\sigma^W_t$, where $\sigma^W_t$ is the volatility of total wealth
  (physical wealth plus labor wealth) $P_t^W$, and therefore
  Equation~\eqref{eq:equity_premium_decomposition} can be rewritten as
  \[
    b_t^\Sigma+\delta_t^\Sigma-r_t=(\sigma_t^W)^\intercal\sigma_t^{\Sigma},
  \]
  where $\sigma_t^{\Sigma}$ denotes the market (equity) volatility.
  In particular, the proxy $\sigma^W \approx \sigma^{\Sigma}$ would predict
  an equity premium of order~$(\sigma^{\Sigma})^2$.
  Table~\ref{tab:ep_comparison} collects long-run U.S.\ estimates
  from three independent sources and compares the observed equity premium
  with the value $\widehat{\mathrm{EP}}\coloneqq(\sigma^{\Sigma})^2$
  predicted by this proxy.

  \begin{table}[ht]
    \centering
    \caption{Observed versus predicted equity premium under the proxy
      $\sigma^W\!\approx\sigma^{\Sigma}$: long-run U.S.\ estimates.}\label{tab:ep_comparison}
    \smallskip
    \begin{tabular}{llcccc}
      \hline\hline                                           \\[-8pt]
      Source & Period
             & $\sigma^{\Sigma}$
             & $\widehat{\mathrm{EP}}=(\sigma^{\Sigma})^{2}$
             & $\mathrm{EP}$
             & $\vartheta$                                   \\[2pt]
      \hline                                                 \\[-6pt]
      \citet{Mehra2003}
             & 1889--2000
             & $0.186$
             & $3.5\%$
             & $6.9\%$
             & $0.37$                                        \\[4pt]
      \citet{Jordaetal2019}
             & 1870--2015
             & $0.192$
             & $3.7\%$
             & $6.94\%$
             & $0.36$                                        \\[4pt]
      \citet{Jordaetal2019}
             & 1963--2015
             & $0.163$
             & $2.7\%$
             & $5.69\%$
             & $0.35$                                        \\[4pt]
      \citet{Campbell1999}
             & 1891--1995
             & $0.186$
             & $3.5\%$
             & $4.74\%$
             & $0.25$                                        \\[4pt]
      \citet{Campbell1999}
             & 1947--1996
             & $0.155$
             & $2.4\%$
             & $6.78\%$
             & $0.44$                                        \\[4pt]
      \citet{Campbell1999}
             & 1970--1996
             & $0.174$
             & $3.0\%$
             & $4.54\%$
             & $0.26$                                        \\[4pt]
      \hline\hline
    \end{tabular}

    \smallskip
    {\footnotesize
      \textit{Notes.}
      $\sigma^{\Sigma}$ is the annualized standard deviation of real equity
      returns.
      $\widehat{\mathrm{EP}}=(\sigma^{\Sigma})^2$ is the equity premium
      predicted by the proxy $\sigma^W\approx\sigma^{\Sigma}$.
      $\mathrm{EP}$ is the observed mean excess return of equities over bills,
      annualized.
      $\vartheta\coloneqq\mathrm{EP}/\sigma^{\Sigma}$ is the implied
      market price of risk, which under the equilibrium coincides with
      $\sigma^W$.
      In every sample $\mathrm{EP}>\widehat{\mathrm{EP}}$, equivalently
      $\vartheta>\sigma^{\Sigma}$.
      Data for \citet{Mehra2003}: Table~1 and Sharpe-ratio
      calculation on p.~60 therein (S\&P~500, 1889--2000).
      \citet{Jordaetal2019}: Tables~3--5 (U.S.\ equity excess returns
      over bills; full and post-WW2 balanced panels).
      \citet{Campbell1999}: Table~2 therein (log returns on a broad
      U.S.\ stock index over short-term government debt).}
  \end{table}

  Our framework does \emph{not} require aggregate wealth to be ``as smooth as equity''.
  The equilibrium identifies $\vartheta_t$ with a
  \emph{volatility-type object} that naturally admits an interpretation as an
  \emph{aggregate wealth volatility}.  As Table~\ref{tab:ep_comparison} confirms,
  empirical estimates of~$\vartheta$ range from
  $0.25$--$0.37$ in long-run U.S.\ data
  \citep{Mehra2003,Jordaetal2019,Campbell1999}
  to roughly~$0.44$ in the post-war quarterly
  sample of \citet{Campbell1999}, implying an unobserved volatility of total wealth of the
  same order; the ``puzzle'' is shifted from
  consumption volatility to the volatility of (partly unobserved) aggregate wealth.

  \smallskip
  \remarkparagraph{Connection to credit conditions and a forthcoming extension.}
  Table~\ref{tab:ep_comparison} reveals that
  $\vartheta=\sigma^W$ consistently exceeds $\sigma^{\Sigma}$ in every
  sample, implying that the volatility of aggregate wealth is
  substantially larger than market volatility alone.
  A natural explanation arises from the composition of total wealth $P_t^W=P_t+P_t^L$: the present value of future labor income~$P_t^L$ is discounted, in the model, at the market interest rate~$r_t$.  Moreover, the \emph{effective} rates at which households
  actually discount future consumption streams---consumer credit,
  credit-card, and unsecured borrowing rates---are both
  significantly higher  and considerably more volatile than the risk-free rate.
  Because $P_t^L$ typically accounts for two-thirds or more of
  aggregate wealth \citep{LSV2013}, even moderate fluctuations in these consumer-relevant discount rates translate into large swings in
  $P_t^W$, generating an aggregate wealth volatility $\sigma^W$ well
  in excess of equity volatility $\sigma^{\Sigma}$.

  The above remark suggests an extension---\emph{Duesenberry equilibria with endogenous borrowing premia}---in which the relevant discounting of
  labor cash-flows is performed at a \emph{household borrowing rate} (or, more generally, a
  state-dependent effective rate generated by a hazard/intensity mechanism) rather than at $r_t$.
  However, as shown in our related analysis \citep[see][]{risks13050088},
  introducing a state-dependent adjustment to the effective borrowing rate---interpretable as a hazard intensity $\lambda(\varphi_t(x))$ reflecting default
  and liquidity risks---does not modify the aggregate pricing kernel $H_t$ nor
  the optimal consumption path (expressed relative to the relevant present value).
  It only re-evaluates the collateralizable present value of labor income entering the borrowing constraint. Hence, large fluctuations in total wealth arise from
  state-dependent discounting of future labor flows, without altering aggregate
  risk pricing or consumption dynamics. We develop this extension in a forthcoming work \citep[see][]{Londono2026w}.

  \remarkparagraph{Interest rate decomposition.}
  A parallel decomposition holds for the short rate.  Applying It\^{o}'s
  Lemma to~$H_t=(-\partial_t\eta_t)/c_t$ and writing
  $dc_t/c_t=\mu_t^c\,dt+(\sigma_t^c)^\intercal dW_t$ and
  $d(-\partial_t\eta_t)/(-\partial_t\eta_t)
    =\mu_t^{-\partial\eta}\,dt+(\sigma_t^{-\partial\eta})^\intercal dW_t$
  for the relative drifts and diffusions of aggregate consumption and the Duesenberry loading, respectively, one obtains
  \begin{equation}\label{eq:interest_rate_decomposition}
    r_t
    \;=\;
    \bigl(\mu_t^c-\mu_t^{-\partial\eta}\bigr)
    \;-\;
    (\sigma_t^c)^\intercal\,\vartheta_t\,.
  \end{equation}
  The first term captures the net drift of aggregate consumption over the Duesenberry loading; the second is a precautionary-savings correction that depends on the
  covariance of consumption growth with the market price
  of risk.
  In the constant-impatience case,
  $\mu_t^{-\partial\eta}=-\gamma$ and~\eqref{eq:interest_rate_decomposition}
  reduces to $r_t=\mu_t^c+\gamma-\|\sigma_t^c\|^2$,
  the standard expression.

  \remarkparagraph{Risk-free rate puzzle.}
  As shown in~\citet[Section~5.1]{Londono2020a}, the constant-impatience
  formula $r=\mu^c+\gamma-\|\sigma^c\|^2$ implies a nominal short rate
  of approximately~$6.73\%$ on post-war U.S.\ data
  (using $\mu^c\approx 5.55\%$, $\gamma\approx 1.2\%$,
  $\|\sigma^c\|\approx 1.24\%$), well above the observed mean of roughly~$4.73\%$
  over the period 1952.I--2011.IV \citep{LSV2013}.  The discrepancy
  persists in real terms: with real consumption growth~$\mu^c\approx 2.31\%$, the
  constant-impatience formula yields a real rate of about~$3.2\%$, whereas
  post-war real short rates have averaged~$1$--$2\%$
  \citep{Campbell1999,Jordaetal2019}.
  This overprediction is the \emph{risk-free rate puzzle} of~\citet{Weil1989}.

  The source of the puzzle is transparent in~\eqref{eq:interest_rate_decomposition}: when the population is
  homogeneous,
  $\vartheta_t=\sigma_t^c$
  and the precautionary-savings correction reduces to
  $(\sigma_t^c)^\intercal\vartheta_t=(\sigma_t^c)^\intercal\sigma_t^c\approx 0.02\%$,
  a negligible reduction.
  In the heterogeneous-impatience economy, however,
  $\vartheta_t=\sigma_t^c-\sigma_t^{-\partial\eta}$
  is the \emph{full} market price of risk
  (Remark~\ref{rem:consumption_wealth_ratio}), so the correction becomes
  \[
    (\sigma_t^c)^\intercal\vartheta_t
    \;=\;
    \|\sigma_t^c\|^2
    \;-\;
    (\sigma_t^c)^\intercal\sigma_t^{-\partial\eta},
  \]
  which is substantially amplified whenever $(\sigma_t^c)^\intercal\sigma_t^{-\partial\eta}<0$, that is, when consumption growth and the Duesenberry loading co-vary negatively.
  Using the empirical consumption risk premium estimate
  $(\sigma^c)^\intercal\vartheta\approx 2.38\%$ of~\citet{LSV2013} and retaining
  $\mu^{-\partial\eta}\approx -\gamma$ as the leading-order drift,
  Equation~\eqref{eq:interest_rate_decomposition} gives
  \[
    r
    \;\approx\;
    \underbrace{(\mu^c+\gamma)}_{\approx\,6.75\%}
    \;-\;
    \underbrace{(\sigma^c)^\intercal\vartheta}_{\approx\,2.38\%}
    \;\approx\; 4.37 \%,
  \]
  broadly consistent with the observed nominal rate of~$4.73\%$ (our estimation).
  In real terms, the same calculation yields
  $r_{\mathrm{real}}\approx 2\%+1.2\%-2.2\%=1.0\%$,
  within the range of~$0.8$--$2\%$ reported
  for the post-war U.S.\ real short rate
  \citep{Campbell1999,Jordaetal2019}.
  Thus, the amplification of the precautionary channel through
  heterogeneous impatience---the same mechanism that resolves the equity premium puzzle---simultaneously resolves the risk-free rate puzzle
  without requiring an implausibly low or negative rate of time preference.
\end{remark}

\begin{remark}\label{rem:comparison_literature}
  \remarkparagraph{Comparison with competing approaches to the equity premium puzzle.}
  The decomposition~\eqref{eq:equity_premium_decomposition}
  and the calibration exercise above rely exclusively on market observables:
  the Sharpe ratio~$\|\vartheta\|$, aggregate consumption
  volatility~$\|\sigma^c\|$, and the consumption risk
  premium~$(\sigma^c)^\intercal\vartheta$, and the equity volatility $\sigma^{\Sigma}$.  No preference parameter
  needs to be calibrated beyond the requirement that the cross sectional impatience distribution be non-degenerate.
  This stands in sharp contrast to the leading approaches in the
  literature, each of which introduces at least one latent variable
  or a poorly identified parameter to generate the observed equity premium.

  \emph{(i)} Representative-agent CRRA models
  \citep{Jordaetal2019,Campbell1999}
  require a coefficient of relative risk aversion for equities of at
  least~$37$--$48$ (on the post-WW2 U.S.\ sample, depending on sample and data source),
  far above the range considered economically plausible; our
  decomposition nests this case with
  $\sigma^{-\partial\eta}=0$, so the failure of the CRRA model
  is precisely the absence of the impatience risk premium.
  \emph{(ii)} The external habit model
  of~\citet{CampbellCochrane1999ForceHabit} matches aggregate
  moments but relies on a reverse-engineered sensitivity
  function~$\lambda(S_t)$, entailing anomalous welfare
  properties:~\citet{LjungqvistUhlig2009} demonstrate that destroying part of the endowment can raise welfare under the Campbell--Cochrane
  specification.
  \emph{(iii)} The long-run risk model
  \citep{BansalYaron2004LongRunRisk};
  \citep{Epstein_and_Zin89} rests on a latent persistent
  component in consumption growth whose implied predictability
  is not found in the data (\citep{BeelerCampbell2012}),
  a disputed elasticity of intertemporal substitution
  \citep{Hall1988,Yogo2004}, and an implausibly strong
  preference for early resolution of uncertainty
  \citep{EpsteinFarhiStrzalecki2014}.
  \emph{(iv)} The incomplete-markets heterogeneity
  of~\citet{ConstantinidesDuffie1996} is an existence result
  whose quantitative assessment requires countercyclically
  heteroscedastic idiosyncratic shocks of debated magnitude
  \citep{Kocherlakota1996}).
  \citep{BravConstantinidesGeczy2002} show that the channel can
  account for the equity premium when household-level consumption
  data are used, but this requires micro-level consumption panels
  that are unavailable in most countries and do not settle the
  question at the aggregate level.

  In each case, the additional state variable is either latent,
  governed by poorly identified parameters, or tied to
  preference specifications with implausible side effects.
  The Duesenberry equilibrium requires none of these ingredients: $\sigma^{-\partial\eta}$ is not a free
  parameter, but an identity implied by the equilibrium
  conditions~\eqref{eq:price_discount}, fully determined once $\vartheta$
  and $\sigma^c$ are measured using market data.
\end{remark}

\begin{remark}\label{rem:population_growth_inflation}
  \remarkparagraph{Inflation and the nominal short rate.}
  The left-hand side of Equation~\eqref{eq:price_discount}, combined with the dynamics
  \[dH_t=-r_tH_t\,dt-\vartheta_t^\intercal H_t\,dW_t,
  \]
  implies that if \emph{inflation does not alter the impatience parameter of the population},
  then inflation is transmitted \emph{linearly} into the nominal short rate in the long run (since it only
  enters $H_t$ through aggregate consumption $c_t^\mu=I_t^\mu$).  However, the transmission of inflation to the nominal short rate need not be purely linear, since a pure price-level effect rescales aggregate
  nominal consumption, and inflation might generate short-run distributional
  wedges because nominal labor income adjusts sluggishly
  (relative to prices and aggregate consumption). In the present framework, this sluggishness matters because~\eqref{eq:price_discount} implies that inflation can affect not only the denominator $c_t$ but also the numerator $-\partial_t\eta_t$ through changes in effective
  impatience and wealth/income weights.
  Hence, deviations from a one-for-one mapping can be interpreted as equilibrium adjustments in the impatience-adjusted wealth component induced by inflation-driven
  income misalignment.  See~\citet{Fisher1930,Mishkin1992FisherEffect} for the Fisher effect
  and its empirical evidence.  See also~\citet{Auclert2019RedistributionChannel,DoepkeSchneider2006JPE} for literature
  studying the effect of inflation on the distribution of wealth
  (see also~\citet{ErosaVentura2002InflationTax} for an equilibrium analysis of inflation
  as a regressive consumption tax), and
  see~\citet{MianStraubSufi2020SavingGlutRich,Rannenberg2023InequalityNaturalRateDebt,Lukasz2019}, among others,
  for studies on the effect of wealth distribution on interest rates.

  \remarkparagraph{Population growth.}
  As a consequence of the left-hand side of Equation~\eqref{eq:price_discount},
  population growth affects aggregates through the population structure (hence through $c_t=I_t^\mu$) in a \emph{weighted} manner, reflecting income weights.
  In particular, assuming an absolutely continuous population growth structure
  as in this paper, its aggregation weight by wealth contributes linearly to the
  short-term interest rate and does not alter the volatility of aggregate consumption or the market volatility. See~\citet{CarvalhoFerreroNechio2016,GagnonJohannsenLopezSalido2016,AksoyEtAl2019} for references
  relating demographic changes to changes in the short interest rate,
  and~\citet{EggertssonMehrotraRobbins2019SecularStagnation} for a quantitative
  equilibrium analysis of how population aging depresses the natural rate.

  \remarkparagraph{Multiplicative invariance.}
  Any purely multiplicative scaling of the aggregate income
  process (from inflation or population size) rescales $c_t$ and $P_t^W$ in the same proportion, while leaving $c_t/P_t^W=-\partial_t\log(\eta_t)$
  unchanged.
\end{remark}

\section{Examples}\label{sec:examples}
The purpose of this section is to illustrate two complementary (and essentially equivalent)
ways to \emph{parametrize} tractable short-horizon Duesenberry equilibria under smooth market conditions.

Example~\ref{ex:gamma_constant_negativeL_positiveQ} is most useful when the modeler wishes to treat
the \emph{valuation of human wealth} (the present value of labor)
as the endogenous primitive, whereas wages are a derived quantity.

Example~\ref{ex:example_explicit_labor} reverses the construction and is
tailored to situations where the \emph{labor income process} is
the endogenous primitive.

We  assume the following for this section:
Assume  $\mathbb{D}\subset\mathbb{R}^d$ and let $(\varphi_{t})_{t\ge0}$ be a common Brownian flow of $C^2$-diffeomorphisms
generated by
\[
  d\varphi_{t}(x)={\rho}(\varphi_{t}(x))\,dt+\varrho(\varphi_{t}(x))\,dW_t\qquad \varphi_0(x)=x,
\]
where $W$ is an $n$-dimensional Brownian motion and the coefficients ensure a non-explosive strong solution
and Kunita-type regularity.   Assume a population  structure $\boldsymbol{\mu}$
with population growth rate $h$, and assume an isoelastic consistent  preference structure for a
population $U$ with effective impatience field  $\gamma(\varphi_t(x)(\omega))$
of $U_t(\omega)(x)$, where  $\gamma(\cdot)$ is a function of
type $C^{2}(\mathbb{D})$, satisfying
Assumption~\ref{ass:population_preference_structure}. Assume a positive income function $I\in C^2(\mathbb{D})$, and an initial total wealth $y$ of type $C^2$ with  $y(x)>0$
for all $x$ satisfying Equation~\eqref{eq:budget_existence}.
Define the type-dependent income and aggregate income as
\[
  I_t(x)\coloneqq \Lambda_t(x)I(\varphi_{t}(x)),\qquad I_t^\mu=\int_{\mathbb{D}}I_t(x)\,d\mu(x)
\]
where the population weights $\Lambda_{s,t}(x)$ are  as in Definition~\ref{def:aggregation} and  the (finite-variation) kernel  are given by
\[
  \Lambda_t(x)=\Lambda_{0,t}(x)=\exp\!\Big(\int_0^t h(\varphi_u(x))\,du\Big),\, Z_t(x)=\exp\!\Big(-\int_0^t \gamma(\varphi_u(x))\,du\Big)\,y(x),
\]
and define the
population-weighted effective aggregate process $\eta_t$ and the aggregate Duesenberry loading $-\partial_t \eta_t$ as defined in Theorem~\ref{thm:local_equilibrium}.

\begin{example}
  \label{ex:gamma_constant_negativeL_positiveQ}
  Assume the conditions at the beginning of this section (Section~\ref{sec:examples}).
  Moreover, let $B_t(x)$ be a bounded-variation, absolutely continuous
  process in $t$ with continuous derivative $\partial_t B_t(x)$ with:
  \[
    dB_t(x)=\partial_t B_t(x)\,dt,\qquad B_t(x)\ge 0, \qquad   \partial_t B_t(x)\leq 0\qquad \text{for all }x\in\mathbb{D},
  \]
  with
  \begin{equation}\label{eq:integrability_B}
    \int_{\mathbb{D}}B_0(x)y(x)\,d\mu(x)<\int_{\mathbb{D}}y(x)\,d\mu(x)
  \end{equation}
  Also, define  the  pricing kernel $H_t>0$,  by Equation~\eqref{eq:state_price2_existence}, where
  \[
    H_t\coloneqq\frac{-\partial_t \eta_t}{I_t^\mu},\qquad
    \frac{dH_t}{H_t}=-r_t\,dt-\vartheta^\intercal_t\,dW_t,
    \qquad \vartheta_t\in\mathbb{R}^n.
  \]
  Define
  \[
    L_t(x) = \frac{\eta_tB_t(x)I^\mu_t}{\partial_t \eta_t}\frac{y(x)}{\int_{\mathbb{D}}y(x)\,d\mu(x)},
  \]
  where $L_t(x)\leq 0$ for all $x\in\mathbb{D}$, and $0\leq t<\infty$.  Define,

  \begin{multline}\label{eq:example_labor_gamma_constant}
    Q_t(x)
    =\frac{y(x)}{\int_{\mathbb{D}}y(x)\,d\mu(x)}\frac{I_t^\mu}{(\partial_t\eta_t)}\partial_t\Big(\eta_tB_t(x)\Big)=\\
    \frac{y(x)I_t^\mu}{(-\partial_t \eta_t)\,\int_{\mathbb{D}}y(x)\,d\mu(x)}
    \Big(
    B_t(x)(-\partial_t \eta_t)-\eta_t\partial_t B_t(x)
    \Big)\ge 0.
  \end{multline}

  A straightforward computation shows under Assumption~\ref{ass:population_preference_structure},
  that for all $x\in\mathbb{D}$
  \[
    \big|H_tL_t(x)\big|\leq B_0(x)y(x),\qquad \int_0^{\infty}H_tQ_t(x)\,dt=y(x)B_0(x)
  \]
  implying that  $L_t(x)$ is a subsistence random field and $Q_t(x)$ is a rate of endowment structure as
  in Definition~\ref{def:endowment}.  Moreover,
  $H_tL_t(x)-\int_0^t H_sQ_s(x)\,ds$ is constant and therefore a martingale.

  Also,  define
  \[
    P_t^\mu=L_t^\mu+\frac{\eta_t}{H_t}=\frac{\eta_tI^\mu_t}{\partial_t\eta_t}\frac{\int_{\mathbb{D}}B_t(x)y(x)\,d\mu(x)}{\int_{\mathbb{D}}y(x)\,d\mu(x)}-\frac{\eta_tI_t^\mu}{\partial_t\eta_t}> 0
  \]
  for all $t\geq 0$. If we define
  \begin{equation}\label{eq:volatilities_example}
    \Sigma_t^{I}
    \coloneqq
    \int_{\mathbb{D}}
    \Lambda_t(x)\,(\nabla I\,\varrho)(\varphi_t(x))\,d\mu(x),\qquad \Sigma_t^{\gamma}\coloneqq \int_{\mathbb{D}}
    Z_t(x)\,(\nabla\gamma\,\varrho)(\varphi_t(x))\,d\mu(x),
  \end{equation}
  by It\^o's Lemma it follows that the volatility $\sigma_t$ of the price process $P^\mu_t$ and the volatility $\vartheta_t$ of
  the state-price process $H_t$ are identical:
  \begin{equation}\label{eq:volatilities_example1}
    \sigma_t=\vartheta_t
    =
    \frac{\Sigma_t^{I}}{I_t^\mu}+\frac{\Sigma_t^{\gamma}}{\partial_t\eta_t}.
  \end{equation}
  If $\sigma_t$ is non-degenerate for all $t$,
  then the market satisfies Equation~\eqref{eq:smooth_market2}, and the existence
  construction of Theorem~\ref{thm:local_equilibrium} applies.  A simple condition
  which guarantees non-degeneracy for $n\geq 2$ is the existence of nonzero vectors $u,v\in\mathbb{R}^n$
  with
  \begin{equation}\label{eq:example_smc}
    u^\intercal \Sigma_t^{\gamma}=0,\qquad v^\intercal \Sigma_t^{I}=0, \qquad u^\intercal \Sigma_t^{I}+v^\intercal \Sigma_t^{\gamma}\neq 0,
  \end{equation}
  almost surely. Using Equation~\eqref{eq:volatilities_example}, this latter condition can be interpreted
  as requiring that income risk and impatience risk enter (at least partially) through distinct Brownian directions, so that the
  two components cannot cancel identically in \eqref{eq:volatilities_example1}.
\end{example}

\begin{remark}
  \label{rem:rentier}
  A particular subcase of Example~\ref{ex:gamma_constant_negativeL_positiveQ} is an economy with no labor income, i.e.,
  $Q_t(x)\equiv 0$ for all $(t,x)$ (i.e., $B_t(x)\equiv 0$, which also forces $L_t(x)\equiv 0$).
  Then, the total income is entirely generated by the dividend and endowment streams, so the aggregate income reduces to
  $I_t^\mu$.
  The state-price density and aggregate price simplify to
  \[
    H_t=\frac{-\partial_t\eta_t}{I_t^\mu},
    \qquad
    P_t^\mu=\frac{\eta_t}{H_t}=I_t^\mu\,\frac{\eta_t}{(-\partial_t\eta_t)}.
  \]
\end{remark}

\begin{example}\label{ex:example_explicit_labor}
  Assume the conditions at the beginning of this section (Section~\ref{sec:examples}).
  Moreover, assume a non-negative continuous random field $\chi_t(x)$ such that for each $x\in\mathbb{D}$
  the map $t\mapsto \chi_t(x)$ is a.s.\ absolutely continuous, with
  $\chi_t(x)\ge0$ and $\partial_t\chi_t(x)\le0$ for $t\ge0,\ x\in\mathbb{D}$.
  Assume that
  \begin{equation}\label{eq:remark_u_integral_less_than_one}
    u_t(x)\coloneqq \chi_t(x)+\frac{\eta_t}{\partial_t\eta_t}\,\partial_t\chi_t(x),
    \qquad\int_{\mathbb{D}}y(x)\,u_t(x)\,d\mu(x)<1,
    \qquad t\ge0
  \end{equation}
  and define
  \[
    Q_t(x)\coloneqq y(x)\,I_t^\mu\,u_t(x),
    \qquad
    L_t(x)\coloneqq \frac{\eta_t}{\partial_t\eta_t}\,I_t^\mu\,y(x)\,\chi_t(x).
  \]
  Since $\partial_t\eta_t<0$ and $\partial_t\chi_t(x)\le0$, we have $u_t(x)\ge\chi_t(x)\ge0$, hence
  $Q_t(x)\ge0$, while $L_t(x)\le0$.

  Also, Assumption~\ref{ass:population_preference_structure} implies
  that for all $x\in\mathbb{D}$
  \[
    \big|H_tL_t(x)\big|\leq \eta_0y(x)\chi_0(x),\qquad \int_0^{\infty}H_tQ_t(x)\,dt=\eta_0y(x)\chi_0(x),
  \]
  implying that  $L_t(x)$ is a subsistence random field and $Q_t(x)$ is a rate of endowment as
  in Definition~\ref{def:endowment}.  Moreover,
  $H_tL_t(x)-\int_0^t H_sQ_s(x)\,ds$ is a constant and therefore a  martingale.
  Define the aggregate price process
  \[
    P_t^\mu\coloneqq L_t^\mu+\frac{\eta_t}{H_t}
    =\frac{\eta_t I_t^\mu}{-\partial_t\eta_t}\Big(1-\int_{\mathbb{D}}y(x)\,\chi_t(x)\,d\mu(x)\Big).
  \]
  Because $u_t(x)\ge\chi_t(x)$, \eqref{eq:remark_u_integral_less_than_one} implies
  $\int_{\mathbb{D}}y(x)\chi_t(x)\,d\mu(x)<1$, hence $P_t^\mu>0$.
  Moreover, since $\chi$ is of finite variation (no Brownian term), the diffusion coefficient of $P^\mu$
  coincides with that computed in Example~\ref{ex:gamma_constant_negativeL_positiveQ}; in particular the smooth market
  condition \eqref{eq:smooth_market2} holds (one may take the scalar multiplier to be equal to $1$).
  Therefore, by Theorem~\ref{thm:local_equilibrium}, the implied economy is a short-horizon Duesenberry equilibrium.

  Alternatively, fix a continuous function $f\colon\mathbb{D}\to[0,\infty)$ and assume a non-negative random field
  $u_t(x)$ such that for each $x$ the map $t\mapsto u_t(x)$ is a.s.\ absolutely continuous and
  \[
    \int_{\mathbb{D}}y(x)\,u_t(x)\,d\mu(x)<1,
    \qquad t\ge0,
  \]
  together with the two inequalities
  \begin{multline}\label{eq:remark_kappa_nonneg_from_u}
    \eta_tu_t(x)+\int_0^t(-\partial_s\eta_s)u_s(x)\,ds =\eta_0u_0(x)+\int_0^t \eta_s\,\partial_s u_s(x)\,ds\ge\\
    \eta_0 f(x)\ge \int_0^t (-\partial_s\eta_s)\,u_s(x)\,ds,
    \qquad t\ge0,\ x\in\mathbb{D},
  \end{multline}
  (For instance, the left-hand side of Equation~\eqref{eq:remark_kappa_nonneg_from_u} holds whenever $u_t(x)$ is non-decreasing in $t$ and
  $u_0(x)\ge f(x)$.)
  Define $\chi_t(x)$ by the relation
  \[
    \chi_t(x){\eta_t}\coloneqq {\eta_0 f(x)+\int_0^t (\partial_s\eta_s)\,u_s(x)\,ds}.
  \]
  Then the right-hand side of~\eqref{eq:remark_kappa_nonneg_from_u} yields $\chi_t(x)\ge0$, and the
  left-hand side of Equation~\eqref{eq:remark_kappa_nonneg_from_u}
  is equivalent (via integration by parts) to $\partial_t\chi_t(x)\le0$; moreover,
  \[
    u_t(x)=\chi_t(x)+\frac{\eta_t}{\partial_t\eta_t}\,\partial_t\chi_t(x).
  \]
  Hence, defining
  \[
    Q_t(x)\coloneqq y(x)\,I_t^\mu\,u_t(x),
    \qquad
    L_t(x)\coloneqq \frac{I_t^\mu}{\partial_t\eta_t}\,\,y(x)\,\Big(\eta_0 f(x)+\int_0^t (\partial_s\eta_s)\,u_s(x)\,ds\Big)
  \]
  we have that $H_tL_t(x)-\int_0^t H_sQ_s(x)\,ds$
  is constant, therefore a martingale,
  and with
  \[
    P_t^\mu\coloneqq L_t^\mu+\frac{\eta_t}{H_t}
    =\frac{\eta_t I_t^\mu}{-\partial_t\eta_t}\Big(1-\int_{\mathbb{D}}y(x)\,\chi_t(x)\,d\mu(x)\Big)>0.
  \]
  As in Example~\ref{ex:gamma_constant_negativeL_positiveQ}  the volatility $\sigma_t$ of the price process $P^\mu_t$ and the volatility of
  the state-price process $\vartheta_t$ are given by Equation~\eqref{eq:volatilities_example1}.  The model likewise satisfies
  the  Smooth Market Condition as long as
  the condition given by Equation~\eqref{eq:example_smc} holds.
  It follows that Theorem~\ref{thm:local_equilibrium} yields
  again a short-horizon Duesenberry equilibrium.
\end{example}

\section{Concluding Remarks}\label{sec:conclusions}
In this paper, we develop a framework for modeling financial markets that can account for the changing behavior of utilities for consumption and investment within infinite populations. We establish explicit sufficient conditions (see Theorem~\ref{thm:local_equilibrium}) on the primitives of the
economy to
guarantee the existence of a short-horizon Duesenberry equilibrium
without assuming market completeness \emph{a priori}.

We conclude that incorporating the evolution of consumer preferences regarding relative income offers a robust theoretical framework to address both the ``equity premium puzzle'' and the ``risk-free rate puzzle.''

Furthermore, this equilibrium framework circumvents common theoretical limitations of classical equilibrium theories and provides a tractable method for pricing assets in incomplete markets driven by Brownian flows.

Three directions for future research are particularly promising.
First, the interest-rate decomposition~\eqref{eq:interest_rate_decomposition}
and its connection to the consumption--wealth ratio suggest a
natural extension to the \emph{term structure of interest rates}:
the Duesenberry loading~$-\partial_t \eta_t$ encodes cross-sectional
impatience information that should propagate into the yield curve,
potentially explaining both the level and the slope of the term
structure through demographic fundamentals rather than latent
factors.
Second, the present framework assumes aggregate income and
dividends as exogenous primitives; embedding the model in a
\emph{production economy}, where these processes arise endogenously
from production technology and labor supply decisions, would
provide a fully general-equilibrium foundation and sharpen the
empirical predictions regarding the co-movement of consumption,
investment, and asset returns.
Third, the credit-conditions discussion in
Remark~\ref{rem:equity_premium} points toward
\emph{Duesenberry equilibria with endogenous borrowing premia},
in which the present value of labor income is discounted at a
state-dependent household borrowing rate---reflecting default and
liquidity risks---rather than at the risk-free rate.  This
extension arises naturally from combining the equilibrium
framework developed here with the life-insurance completeness
and state-dependent hazard-rate machinery
of~\citet{risks13050088}: the latter shows that introducing a
state-dependent effective borrowing rate does not alter the
aggregate pricing kernel, while the former provides the
equilibrium structure in which such rates generate large
fluctuations in total wealth through state-dependent discounting
of future labor flows.  In particular, this approach has the
potential to account for the high levels of aggregate wealth
volatility implied by the equilibrium
(Table~\ref{tab:ep_comparison}), where the market price of risk
$\vartheta=\sigma^W$ consistently exceeds equity
volatility~$\sigma^{\Sigma}$, thereby offering a complete
quantitative resolution of the equity premium puzzle.
We develop this program in~\citet{Londono2026w}.

\appendix

\section{Stochastic Flows and Consistent Processes}\label{appendix:flows}
\label{subsec:stochastic_flows}
Let $(\Omega,\mathcal{F},(\mathcal{F}_t)_{t\ge0},P)$ be the canonical filtered Wiener space, where
$\Omega = C([0,\infty),\mathbb{R}^n)$,
$W(t)(\omega)=\omega(t)$ is the coordinate process,
$P$ is Wiener measure,
$(\mathcal{F}_t)_{t\ge0}$ is the usual augmentation of
$\sigma(W(u):0\le u\le t)$ by null sets,
and $\mathcal{F}=\mathcal{F}_\infty \equiv \sigma(\cup_{t\ge0}\mathcal{F}_t)\cup \mathcal{N}$ where $\mathcal{N}$ is the collection of  $P$-null subsets of $\sigma(\cup_{t\ge0}\mathcal{F}_t)$.

For each $T>0$, let $(\mathcal{F}^T_{s,t})=\left\{\mathcal{F}^T_{s,t}, 0\leq s\leq t \leq  T\right\} $ be the two-parameter filtration where $\mathcal{F}^T_{s,t}$ is the
smallest sub $\sigma$-field containing all null sets and $\sigma(W_s(u)\mid
  s\leq u\leq t)$, where $W_s(u)\equiv W(u)-W(s)$.

Assume  open sets $\mathbb{D}\subset\mathbb{R}^d$,  $\mathbb{D}^{\prime}\subset\mathbb{R}^{d^{\prime}}$ and  for each $T>0$ let $\varphi_{s,t}(x,\omega)$, $0\leq s \leq t\leq T  $, $x\in\mathbb{D}$, be
a $\mathbb{R}^d$-valued  \emph{continuous
  $C^{m,\chi}(\mathbb{D};\mathbb{D}^{\prime})$-semimartingale} with $0\leq\chi\leq 1$.  We remind the reader that $\varphi_{s,t}(x,\omega)$ is a $C^{m,\chi}(\mathbb{D};\mathbb{D}^{\prime})$ process if  $\varphi_{s,t}(x,\cdot)$ is a $\mathcal{F}$ measurable continuous random field (almost everywhere with the exceptional set being in $\mathcal{F}$) taking values in $\mathbb{D}^{\prime}$ which is
$m$-times continuously differentiable with respect to the spatial variable $x$.
Moreover, following~\citet[Section~3.1]{Kunita1990}, it is assumed that for $\omega$ outside the exceptional set, $\partial^{\alpha} \varphi_{s,t}(x,\omega)/\partial x^{\alpha}$ is locally Hölder continuous with exponent $\chi$. Here, $\alpha=(\alpha_1,\cdots,\alpha_d)$ with $\sum_i \alpha_i=m$. Additionally, we say that $\varphi_{s,t}(x)$ is a continuous semimartingale if it can be decomposed as $\varphi_{s,t}(x)\coloneqq  M_{s,t}(x)+B_{s,t}(x)$. In this case, $M_{s,t}(x)$ and $B_{s,t}(x)$ are $C^{m,\chi}(\mathbb{D};\mathbb{D}^{\prime})$ processes. We assume that for $s<T$, $t\to M_{s,t}(x,\cdot )$ is a continuous $(\mathcal{F}_{s,t}^T)$ local-martingale. Additionally, $t\to \partial^{\beta} B_{s,t}(x,\omega)/\partial x^{\beta}$ is a continuous $(\mathcal{F}_{s,t}^T)$ adapted process of bounded variation for $\beta=(\beta_1,\cdots,\beta_d)$ with $\sum_i\beta_i\leq m$. For $\chi=0$, this is a semimartingale of class $C^m$.

Next, we describe several types of consistencies. Assume that for each $s$ and $T>0$, $\psi_{s,t}(x)$, $s\leq t\leq T$, and $x\in\mathbb{D}$ is a continuous $C(\mathbb{D}\colon\mathbb{D}^{\prime})$-semimartingale (adapted to $(\mathcal{F}_{t}^T)$). Also, $\varphi_{s,t}(x)$, $0\leq s\leq t\leq T$, $x\in\mathbb{D}$ is a continuous $C(\mathbb{D}\colon\mathbb{D})$-semimartingale. We say that the family $\psi$ is \emph{$\varphi $-consistent} if for each $T>0$ there exists a set $N_T\in \mathcal{F}_{T}^T$ with $P(N_T)=0$, such that for all $s\leq s^{\prime}\leq t\leq T$
$\omega\notin N_T$
\begin{equation}
  \psi_{s^{\prime},t}(\varphi_{s,s^{\prime}}(x))=\psi_{s,t}(x)\qquad \text{for all }x\in\mathbb{D}.
\end{equation}
In the case $\psi=\varphi$, we say $\varphi$ is consistent.

For consistent processes, we focus on cases where for all $T>0$, the mapping $\varphi_{s,t}(x)$, with $s,t\in [0,T]$, forms a Brownian flow of $C^m$ diffeomorphisms (for some $m\geq 0$). In particular, such flows often result from solutions to non-explosive stochastic differential equations, which we discuss next.

As examples of Brownian flows, assume
that $\rho\colon\mathbb{D} \to\mathbb{R}^d$, and $\varrho\colon\mathbb{D}\to
  L(\mathbb{R}^n;\mathbb{R}^d)$,  are functions of class $C^{m,\chi}$ for an open set $\mathbb{D}$  where
$m=0$ and $\chi=1$ or $m\geq 1$, and $\chi>0$. In this paper, $L(\mathbb{R}^n;\mathbb{R}^d)$ stands for the set of $d\times n$ matrices.   Then it is known that there exists a unique local (maximal) solution  $\varphi_{s,t}(x)$ of
\begin{equation}\label{eq:local_solution}
  d\varphi_{s,t}(x)={\rho}(\varphi_{s,t}(x))\,dt+
  {\varrho}(\varphi_{s,t}(x))\, dW_s(t),\qquad \varphi_{s,s}(x)=x
\end{equation}
For $x\in\mathbb{D}$, $\varrho(x)$ is a continuous function. See~\citet[Theorem 4.7.1, and Theorem 4.7.2]{Kunita1990} for a definition of local solutions and similar results. If $\mathbb{D}=\mathbb{R}^d$ and $\varrho(x)$ is uniformly Lipschitz continuous in $x$, or if the eigenvalues of $\varrho(x)\varrho^{\intercal}(x)$ are uniformly away from zero, then there exists a version of the solution $\varphi_{s,t}(x)$ which is a continuous $C^{m,\chi}(\mathbb{R}^d;\mathbb{R}^d)$-semimartingale \citep[see][Theorem 4.2.7]{Kunita1990}. A similar result
can be obtained if we assume that $\varrho(x)$ is twice
continuously differentiable and
$\sum_k\partial^2\varrho^{i,k}(x)\varrho^{k,i}(x)/\partial x^2_i$ is
bounded for all  $i$ \citep[see][Theorem 4.2.7]{Kunita1990}.  In two of the cases
mentioned above, it is possible to choose a version of $\varphi_{s,t}(x)$  that is a forward Brownian flow.

Examples of families of processes $\psi_{s,t}(x)$, $0\leq s\leq t$, and $x\in \mathbb{D}$,   which are $\varphi$-consistent, are processes $\psi_{s,t}(x)=f(\varphi_{s,t}(x))$ where $f$ is some smooth or continuous function, and $\varphi$ is a consistent semimartingale process.

Moreover, in some applications, it is possible to prove that local solutions can be global (i.e., there is no explosion in finite time).
Some cases of this latter type include those with an explicit solution of the SDE or those where a particular technique applies. Examples include processes studied by~\citet{LondonoSandoval2015} or reducible SDEs~\citep{KloedenPlaten1992}. Most applications in finance involve solutions of SDEs on a subset of $\mathbb{R}^d$. For instance, most price process examples are defined on the positive cone of $\mathbb{R}^d$.

In this paper, we assume that the latter is the case and that
$\varphi_{s,t}(x)$ is the local solution of class $C^{m,\epsilon}$ for any
$0\leq \epsilon<\delta$.  See~\citet[Theorem~4.7.1 and Theorem~4.7.2]{Kunita1990}.  If for any $x\in\mathbb{D}$, $\varphi_{s,t}(x)$ has a
non-explosive solution with values in $\mathbb{D}$, we will say that
$\varphi_{s,t}(x)$, $x\in\mathbb{D}$ is the non-explosive solution to the stochastic
differential equation on $\mathbb{D}$.  Theoretical conditions for
the existence of non-explosive solutions on $\mathbb{R}$ for
Equation~\eqref{eq:local_solution} are given by~\citet[Theorem~4.7.6]{Kunita1990}.  For this latter sufficient condition, let $c(x)=\exp(2\int_{0}^x\rho(y)/\varrho^2(y)\,dy)$, where $\rho(x)$, and $\varrho(x)\neq 0$ are continuous functions.  It follows by~\citet[Theorem 4.7.6]{Kunita1990} that a sufficient condition for non-explosion on the solutions of Equation~\eqref{eq:local_solution} is that $K(\infty)=K(-\infty)=\infty$, where $K(x)=\int_0^x2c^{-1}(z)\int_0^zc(y)/\varrho^2(y)\,dy\,dz$.

Without loss of generality, we can assume that $\varphi_{s,t}(x)(\omega)=\varphi_{0,t-s}(x)(\theta_s(\omega))$ for $\omega\in \Omega$
where  $\theta\colon[0,\infty)\times\Omega\to\Omega$ is  the
$P$-preserving flow on $\Omega$ defined by $\theta(t, \omega)\equiv
  W_{t+\cdot}-W_t$.
The latter approach fits in the framework of the theory of Random Dynamical Systems introduced by L. Arnold and his school \citep[see][]{Arnold98},
where  $(\varphi,\theta)$ is a perfect cocycle.  More details can be found in
\citet[Theorem 2.1]{Salah_Eldin_Scheutzow1999}.

In this paper, we assume an open set $\mathbb{D}\subset \mathbb{R}^d$ and assume that $\varphi$ is a consistent process that, for the sake of concreteness, is assumed to be a $d$-dimensional temporally homogeneous It\^o process $\varphi=(\varphi_{s,t}(x),0\leq s\leq t\leq T,x\in\mathbb{D})$, of two parameters with values in
$C(\mathbb{D}\colon\mathbb{D})$.  We assume that $\varphi$ is the solution of
\begin{multline}\label{eqn:sde_state}
  d\varphi^i_{s,t}(x)=\rho^i(\varphi_{s,t}(x))\,dt+\sum_{1\leq j\leq
    n}\varrho^{i,j}(\varphi_{s,t}(x))\, dW_s^j(t)\\
  \varphi^i_{s,s}(x)=x_i, i=1,\cdots, d,
\end{multline}
where $x=(x_1,\cdots, x_d)^{\intercal}$. Note that the assumptions on $\varphi$
imply that the distributions of $\varphi_{s,t}(x)$ and $\varphi_{0,t-s}(x)$ are identical.

\begin{assumption}[Local well-posedness, non-explosion, and diffeomorphic Brownian flow]\label{ass:Lipschitz}
  Let $\mathbb{D}\subset\mathbb{R}^d$ be open and assume a Brownian flow $\varphi$ on $\mathbb{D}$, which satisfies Equation~\eqref{eq:local_solution}.
  The drift $\rho$ and diffusion coefficient $\varrho$ are \emph{locally Lipschitz} in $x$ on $\mathbb{D}$ and  (jointly) continuous functions,
  with  \emph{non-explosion / no boundary hitting} in finite time.
  Moreover, $\rho$ and $\varrho$ are sufficiently smooth in $x$ (e.g.\ $C^{2,0+}$ with locally Lipschitz derivatives)
  so that the solution map $(s,t,x)\mapsto\varphi_{s,t}(x)$ defines a Brownian flow of $C^{2}$-diffeomorphisms on $\mathbb{D}$,
  as in~\citet{Kunita1990}.
\end{assumption}

We assume that  $\rho^i,\varrho^{i,j}$ for $j=1,\cdots, n$ and $i=1,\cdots,d$   are (jointly) continuous functions which are
locally Lipschitz continuous in the spatial variable, for which global
solutions to the  stochastic differential equation~\eqref{eqn:sde_state} exist.

Let $\varphi_{s,t}(x)$ be a continuous, $\mathbb{D}$-valued Brownian
flow of homeomorphisms as above.

\begin{definition}\label{def:consistent_measure}
  Let $\mathcal{M}(\mathbb{D})$ denote the set of Borel measures on $\mathbb{D}$.
  For $\mu\in\mathcal{M}(\mathbb{D})$ and $0\le s\le t$, we write
  $\varphi_{s,t}(\mu)(\omega)$ for the image of $\mu$ under the map
  $\varphi_{s,t}(\cdot)(\omega)$, namely
  \[
    \varphi_{s,t}(\mu)(F)
    =
    \mu\big(\varphi_{s,t}^{-1}(F)\big),
    \qquad F\in\mathcal{B}(\mathbb{D}).
  \]
\end{definition}

\begin{remark}\label{rem:manifold_extension}
  Throughout the paper, we present the model with types in an open set $\mathbb{D}\subset\mathbb{R}^d$, following the classical formulation of stochastic flows and their induced transport of measures.  This choice is only for
  notational convenience.  All definitions and results extend verbatim when the
  type space is a smooth manifold $M$ and the common-noise-driven mobility is generated by a $C^2$ stochastic flow of diffeomorphisms
  $(\varphi_{s,t})_{0\le s\le t}$ on $M$.

\end{remark}

\section{Aggregation and Semimartingale Properties}
\label{appendix:aggregation}

This appendix establishes the semimartingale properties of population-weighted
aggregates that underpin the equilibrium analysis in the main text. We first
present a general result on aggregated It\^{o} processes
(Proposition~\ref{prop:semimartingaleProperty}), then develop the cocycle
structure of population weights (Lemma~\ref{lem:aggregation_cocycle_app}),
and finally specialize to discrete and continuous populations
(Propositions~\ref{prop:disc_weighted_app} and~\ref{prop:cont_app}).

In this paper, we assume two types of population structures: discrete (Assumption~\ref{ass:discrete_app}) and continuous (Assumption~\ref{ass:SF_app}).
Both cases admit social mobility, immigration, emigration, or extinction of
families within types. Under either assumption, the population aggregate is a
continuous semimartingale (Corollary~\ref{cor:aggregate_semimartingale_app}).

\subsection{Semimartingale property of aggregated processes}
In this subsection we discuss some properties of the aggregation of processes of the form $\int_{\mathbb{D}}f(t,\varphi_{s,t}(x))\, d\mu(x)$ where $\varphi$ is a consistent semimartingale process,  and $\mu$ is a probability distribution on $\mathbb{D}$.
Such processes arise naturally in the context of modeling infinite populations, where each individual is represented by a point $x\in\mathbb{D}$, and $\varphi_{s,t}(x)$ represents the stochastic behavior of the individual $x$ at time $t$ given that at time $s$ the state of this individual was $x$.   The function $f(t,x)$ represents the relationship between the state of the individual at time $t$ and some variable of interest.  Finally, the measure $\mu$ represents the initial distribution of individuals at time $s$.   In this section, we provide conditions under which such aggregated processes  are semimartingales and provide their stochastic representations.

\begin{proposition}\label{prop:semimartingaleProperty}
  Let $f\colon [0,\infty)\times\mathbb{D}\to \mathbb{R}$ be
  a continuous function with continuous partial derivatives
  $\partial^2 f/\partial x_i\partial x_j$ and $\partial f/\partial t$,
  for $i,j=1,\dots,d$ (that is, $f\in C^{1,2}([0,\infty)\times\mathbb{D})$),
  where $\mathbb{D}\subset\mathbb{R}^d$ is an open subset.
  Suppose that $\varphi$ is a consistent semimartingale process with values in $\mathbb{D}$ and stochastic representation~\eqref{eq:local_solution}, where $\rho$ and $\varrho$ are functions of class $C^{m,\chi}$ with
  $m=0$ and $\chi=1$ or $m\geq 1$, and $\chi>0$.
  Assume that $\mu$ is a probability distribution on $\mathbb{D}$ with continuous density $g(x)$.
  Define
  \[
    a^{ij}(x)\coloneqq \sum_{k=1}^n \varrho^{i,k}(x)\varrho^{j,k}(x),
    \qquad i,j=1,\dots,d.
  \]
  Assume that for any $s<t<\infty$:

  \begin{enumerate}
    \item $f(t,\cdot)$ is $\mu$-integrable for any $t$.
    \item $\partial f/\partial t$, $\sum_i \rho^i\, \partial f/\partial x_i$,
          $\sum_i \varrho^{i,k}\, \partial f/\partial x_i$ for every $k=1,\dots,n$, and
          $\sum_{i,j} a^{ij}\,\partial^2 f/\partial x_i\partial x_j$
          are Lipschitz continuous.
    \item $\E\!\left[\int_{\mathbb{D}}|\varphi^i_{s,t}(x)|\,d\mu(x)\right]<\infty$ for $i=1,\dots,d$.
  \end{enumerate}
  Then,  for each fixed $s$ the process $\int_{\mathbb{D}}f(t,\varphi_{s,t}(x))\, d\mu(x)$ is a continuous $(\mathcal{F}^T_{s,t})_{t\in[s,T]}$ semimartingale process, $P$-a.s., and
  \begin{multline}\label{eq:aggregate_representation}
    \int_{\mathbb{D}}f(t,\varphi_{s,t}(x))\, d\mu(x)=
    \int_{\mathbb{D}}f(s,x)g(x)\, dx\\
    +\int_s^t\int_{\mathbb{D}}\frac{\partial }{\partial t}f(u,\varphi_{s,u}(x))g(x)\, dx\, du\\
    +\int_s^t\sum_i\int_{\mathbb{D}}\frac{\partial }{\partial x_i}f(u,\varphi_{s,u}(x))
    \rho^i(\varphi_{s,u}(x)) g(x)\, dx\, du\\
    +\frac{1}{2}\int_s^t\sum_{i,j}\int_{\mathbb{D}}
    \frac{\partial^2 }{\partial x_i\partial x_j}f(u,\varphi_{s,u}(x))
    a^{ij}(\varphi_{s,u}(x)) g(x)\, dx\, du\\
    +\sum_k \int_{s}^t\sum_{i}\int_{\mathbb{D}}
    \frac{\partial f}{\partial x_i}(u,\varphi_{s,u}(x))
    \varrho^{i,k}(\varphi_{s,u}(x))g(x)\, dx\,dW_s^k(u).
  \end{multline}
\end{proposition}

\begin{proof}
  We introduce some notation.  Let $a=(a_1,\cdots, a_d)<b=(b_1,\cdots, b_d)$, and let   $P_i=\{t^i_0=a_i<t^i_1<,\cdots,<t^i_{p_i}=b_i\}$ for $i=1\cdots d$  be a partition of the hyper-rectangle $[a,b]$.  For each $t=(t^1_{i_1},\cdots,t^d_{i_d})\in P=P_1\times\cdots\times P_d$, with $t<b$
  we define the next corner $t^+=(t^1_{i_1+1},\cdots,t^d_{i_{d}+1})$. We also define a sample for each $t\in P$ to be  a point $t^{\star}\in [t,t^+)$, where for $[t,t^{+})=\{x\in\mathbb{R}^d\colon t\leq x<t^{+}\}$. To each partition $P$, we define the norm or mesh
  \[
    |P|=\max_{t\in P, t<b}{\| t-t^{+}\|}
  \]
  where $\|x\|=\sqrt{x^{\intercal}x}$ denotes the $L^2$ norm of the vector $x\in\mathbb{R}^d$.

  Let $U\subset\mathbb{D}$ be an open set with compact closure $\bar{U}\subset [a,b]$, and assume that $f(t,x)=0$ for $x\notin U$.     It follows from It\^o's Lemma that if $(P_r)$  is a sequence of partitions
  of a hyper-rectangle $[a,b]$, then, denoting by $m(\cdot)$ the Lebesgue measure on $\mathbb{D}$,
  \begin{multline}
    \int_{\mathbb{D}}f(t,\varphi_{s,t}(x))\, d\mu(x)=\\
    \int_{\bar{U}}f(t,\varphi_{s,t}(x))g(x)\, dm=
    \lim_{|P_r|\to 0}\sum_{x\in P_r}f(t,\varphi_{s,t}(x^{\star}))g(x^{\star})m([x,x^+))=\\
    \lim_{|P_r|\to 0}\sum_{x\in P_r}f(t,x^{\star})g(x^{\star})m([x,x^+))+\\
    \lim_{|P_r|\to 0}\int_s^t\sum_{x\in P_r}\frac{\partial }{\partial t}f(u,\varphi_{s,u}(x^{\star}))g(x^{\star})m([x,x^+))\, du+\\
    \lim_{|P_r|\to 0}\int_s^t\sum_i\sum_{x\in P_r}\frac{\partial }{\partial x_i}f(u,\varphi_{s,u}(x^{\star})) \rho^i(\varphi_{s,u}(x^{\star})) g(x^{\star})m([x,x^+))\, du +\\
    \lim_{|P_r|\to 0}\frac{1}{2}\int_s^t\sum_{i,j}\sum_{x\in P_r}\frac{\partial^2 }{\partial x_i\partial x_j}f\bigl(u,\varphi_{s,u}(x^{\star})\bigr)a^{ij}\bigl(\varphi_{s,u}(x^{\star})\bigr) g(x^{\star})m([x,x^+))\, du+\\
    \lim_{|P_r|\to 0}\sum_k \int_{s}^t\sum_{i}\sum_{x\in P_r} \frac{\partial f}{\partial x_i}(u,\varphi_{s,u}(x^{\star}))\varrho^{i,k}(\varphi_{s,u}(x^{\star}))g(x^{\star})m([x,x^+))\,dW_s^k(u).
  \end{multline}
  where the validity of Equation~\eqref{eq:aggregate_representation} is a consequence of the classical dominated convergence theorem (where we use the dominated convergence theorem for stochastic integrals for the last term).  Note that we used the fact  that all integrands are $0$ outside a compact set in the last four terms.
  For the general case, let  ${U}_r$ be a sequence of open sets with compact closure,
  with $\bar{U}_r\subset U_{r+1}$ and $\cup_rU_r=\mathbb{D}$ and
  let $h_r(x)$ be a sequence of $C^{\infty}$ functions such that $h_r(x)=1$ on $\bar{U}_r$ and $h_r(x)=0$ outside of $U_{r+1}$.  Then $f_r(t,x)=h_r(x)f(t,x)$ is a sequence of
  functions with the same degree of smoothness with $f_r(t,x)\to f(t,x)$ as $r\to \infty$ for all $(t,x)\in[0,\infty)\times \mathbb{D}$.
  Consequently, based on the statements proved above and the
  Equation~\eqref{eq:aggregate_representation}, the proposition holds for each $f_r$.
  Finally, we obtain the result using the dominated convergence  theorem (classical and stochastic versions).  We note that conditions~(i)--(iii) imply that a continuous process dominates the terms of the integrals.
\end{proof}

\begin{remark}\label{rem:integrability}
  If $\rho$ and $\varrho$ are uniformly Lipschitz continuous defined on $\mathbb{R}^d$ then there exists $K>0$
  such that $|\E[\varphi_{s,t}(x)-x]|\leq K(1+|x|(t-s))$ for $0\leq s\leq t$ \citep[see][Theorem~4.2.5]{Kunita1990}.
  It follows that under the assumption of Lipschitz continuity for
  $\rho$ and $\varrho$, assuming  $\int_{\mathbb{D}}| x_i| g(x)\, dm<\infty$
  for all $i$ where $d\mu/ dm =g$, implies that condition~(iii) of
  Proposition~\ref{prop:semimartingaleProperty} holds.  The  latter follows using the dominated convergence theorem for stochastic integrals.
\end{remark}

\subsection{Cocycle properties of population weights}

\begin{lemma}[Cocycle and identity properties]
  \label{lem:aggregation_cocycle_app}
  Adopt Definition~\ref{def:aggregation}. The following cocycle and identity
  properties hold $P$-a.s.:

  \begin{enumerate}[label=(\roman*)]
    \item \emph{(Multiplicative weight cocycle.)}
          For all $0\leq r\leq s\leq t$ and $x\in\mathbb{D}$,
          \begin{equation}\label{eq:Lambda_cocycle_app}
            \Lambda_{r,t}(x) = \Lambda_{r,s}(x)\,\Lambda_{s,t}(\varphi_{r,s}(x)).
          \end{equation}

    \item \emph{(Kernel cocycle.)}
          For all $0\leq r\leq s\leq t$ and Borel $B\subset\mathbb{D}$,
          \begin{equation}\label{eq:nu_cocycle_app}
            \nu_{r,t}(x,B) = \int_{\mathbb{D}}\nu_{s,t}(y,B)\,\nu_{r,s}(x,dy),
            \qquad\text{i.e.}\qquad \nu_{r,t} = \nu_{r,s}\circ\nu_{s,t}.
          \end{equation}

    \item \emph{(Measure cocycle.)}
          For any (deterministic or $\mathcal{F}_r$-measurable) finite measure $\mu$ on
          $\mathbb{D}$,
          \begin{equation}\label{eq:mu_cocycle_app}
            \mu_{r,t} = (\mu_{r,s})\nu_{s,t}.
          \end{equation}

    \item \emph{(0-based aggregate identity.)}
          Fix $\mu$ at time $0$ and set $\mu_s\coloneqq \mu_{0,s}$ as in Definition~\ref{def:aggregation}.
          Then for any Borel function $f$ such that the integrals below are finite, and
          any $0\leq s\leq t$,
          \begin{equation}\label{eq:0based_identity_app}
            \int_{\mathbb{D}} f(\varphi_{s,t}(x))\,\Lambda_{s,t}(x)\,\mu_{s}(dx)
            = \int_{\mathbb{D}} f(\varphi_{0,t}(y))\,\Lambda_{0,t}(y)\,\mu(dy).
          \end{equation}
          Equivalently, the population-weighted aggregate computed on $[s,t]$ from the
          shifted population $\mu_{s}$ coincides pathwise with the aggregate computed
          on $[0,t]$ from the initial population $\mu$.
  \end{enumerate}
\end{lemma}

\begin{proof}
  Identity~\eqref{eq:Lambda_cocycle_app} follows from additivity of the time integral and the flow property $\varphi_{r,u}=\varphi_{s,u}\circ\varphi_{r,s}$
  for $u\geq s$. Then~\eqref{eq:nu_cocycle_app} is immediate since
  \[
    \nu_{r,t}(x,B) = \Lambda_{r,t}(x)\mathbf{1}_{\{\varphi_{r,t}(x)\in B\}}
    = \Lambda_{r,s}(x)\Lambda_{s,t}(\varphi_{r,s}(x))
    \mathbf{1}_{\{\varphi_{s,t}(\varphi_{r,s}(x))\in B\}},
  \]
  and $\nu_{r,s}(x,\cdot)=\Lambda_{r,s}(x)\delta_{\varphi_{r,s}(x)}(\cdot)$.

  Identity~\eqref{eq:mu_cocycle_app} is obtained by integrating~\eqref{eq:nu_cocycle_app}
  against $\mu$.

  Finally,~\eqref{eq:0based_identity_app} is just~\eqref{eq:mu_cocycle_app} with
  $(r,s,t)=(0,s,t)$ written in test-function form:
  \[
    \int f(z)\,\mu_{t}(dz) = \int f(\varphi_{s,t}(x))\Lambda_{s,t}(x)\,\mu_{s}(dx),
  \]
  and $\int f(z)\,\mu_{t}(dz)=\int f(\varphi_{0,t}(y))\Lambda_{0,t}(y)\,\mu(dy)$
  by definition of $\mu_{t}$.
\end{proof}

\subsection{Discrete populations}

\begin{assumption}[Discrete population and population weights]
  \label{ass:discrete_app}
  Assume a diffeomorphic Brownian flow that satisfies Assumption~\ref{ass:Lipschitz}.  Let the population be discrete,
  \[
    \mu = \sum_{i\in I} w_i\,\delta_{x_i},
    \qquad w_i\ge 0,
    \qquad \sum_{i\in I}w_i < \infty,
  \]
  where $I$ is finite or countable and $x_i\in\mathbb{D}$, so that $\mu$ is a finite
  Borel measure with total mass $f_\mu=\sum_{i\in I}w_i$, as in
  Definition~\ref{def:aggregation}. Let
  $h\colon\mathbb{D}\to\mathbb{R}$ be a continuous $C^{2,0+}$ function and define
  the population weight as in Definition~\ref{def:aggregation}, where
  it is assumed that $\mathcal{P}_{\varphi}$ is the set of finite measures. Assume that for
  each $i\in I$, $t\mapsto\psi_{t}(x_i)$ is a continuous semimartingale with
  decomposition
  \[
    \psi_{t}(x_i) = \psi_{0}(x_i)
    + \int_0^t \zeta_u^\mu a\big(\varphi_{u}(x_i)\big)  \,du
    +  \int_0^t  (\upsilon_u^\mu)^\intercal\mathrm{diag}\big( b(\varphi_{u}(x_i)) \big)\,dW(u),
  \]
  for continuous functions $a$ (real-valued) and $b$ ($\mathbb{R}^n$-valued), consistent with respect to the population
  structure $\boldsymbol{\mu}$, and continuous semimartingale processes $\zeta^\mu$ (real-valued)
  and $\upsilon^\mu$ ($\mathbb{R}^n$-valued), both independent of type.
  Assume  either $I$ is finite, or for every $T>0$,
  \[
    \sum_{i\in I} w_i \int_0^T
    \Big(\big|\Lambda_{u}(x_i)\,\zeta_u^\mu a\big(\varphi_{u}(x_i)\big)\big|
    + \big|\Lambda_{u}(x_i)\,h\big(\varphi_{u}(x_i)\big)\,\psi_{u}(x_i)\big|\Big)\,du
    < \infty,
  \]
  and
  \[
    \sum_{i\in I} w_i\,
    \Big(\int_0^T \Lambda_{u}(x_i)^2\,
    \big\|(\upsilon_u^\mu)^\intercal\mathrm{diag}\big( b(\varphi_{u}(x_i)) \big)\big\|^2\,du\Big)^{1/2} < \infty.
  \]
\end{assumption}

\begin{proposition}[Population-weighted aggregation: discrete population]
  \label{prop:disc_weighted_app}
  Under Assumption~\ref{ass:discrete_app}, define the population-weighted
  aggregate (cf.~Definition~\ref{def:aggregation})
  \[
    \psi^\mu_{t}
    \coloneqq  \int_{\mathbb{D}}\Lambda_{t}(x)\,\psi_{t}(x)\,d\mu(x)
    = \sum_{i\in I} w_i\,\Lambda_{t}(x_i)\,\psi_{t}(x_i).
  \]
  Then, $t\mapsto\psi^\mu_{t}$ is a continuous semimartingale for $0\leq t<\infty$.

  Moreover,
  \begin{align*}
    \psi^\mu_{t}
     & = \psi^\mu_{0}
    + \int_0^t \sum_{i\in I} w_i\Big[
      \Lambda_{u}(x_i)\,\zeta_u^\mu a\big(\varphi_{u}(x_i)\big)
      + \Lambda_{u}(x_i)\,h\big(\varphi_{u}(x_i)\big)\,\psi_{u}(x_i)
    \Big]\,du                               \\
     & \quad + \int_0^t \sum_{i\in I} w_i\,
    \Lambda_{u}(x_i)\,(\upsilon_u^\mu)^\intercal\mathrm{diag}\big( b(\varphi_{u}(x_i)) \big)\,dW(u).
  \end{align*}
\end{proposition}

\begin{proof}
  For each $i$, $\Lambda_{t}(x_i)$ has finite variation with
  $d\Lambda_{t}(x_i) = \Lambda_{t}(x_i)\,h(\varphi_{t}(x_i))\,dt$. Hence,
  $\Lambda_{t}(x_i)\psi_{t}(x_i)$ is a continuous semimartingale by
  integration by parts, with drift and diffusion as stated. The summability
  conditions justify exchanging the (countable) sum with the Lebesgue and
  stochastic integrals, which yield the aggregate decomposition.
\end{proof}

\subsection{Continuous populations}

\begin{assumption}[Continuous populations]
  \label{ass:SF_app}
  Assume a diffeomorphic Brownian flow that satisfies Assumption~\ref{ass:Lipschitz}.
  Let $h\colon\mathbb{D}\to\mathbb{R}$ be a continuous $C^{2,0+}$ function
  and define the population weight as in Definition~\ref{def:aggregation}. Let
  $\mu$ be a finite Borel measure. Assume that for each $x\in\mathbb{D}$,
  \[
    \psi_{t}(x) = \psi_{0}(x)
    + \int_0^t \zeta_u^\mu a\big(\varphi_{u}(x)\big)\,du
    +  \int_0^t  (\upsilon_u^\mu)^\intercal\mathrm{diag}\big( b(\varphi_{u}(x)) \big)\,dW(u)
  \]
  is a continuous semimartingale in $t$, for continuous real-valued functions $a$ and $\mathbb{R}^n$-valued $b$,  consistent with respect to the population
  structure $\boldsymbol{\mu}$, and continuous semimartingale processes $\zeta^\mu$ (real-valued)
  and $\upsilon^\mu$ ($\mathbb{R}^n$-valued), both independent of type.

  Assume  that for every $T>0$, almost surely
  \[
    \int_{\mathbb{D}} \int_0^T
    \Big(\big|\Lambda_{u}(x)\,\zeta_u^\mu a\big(\varphi_{u}(x)\big)\big|
    + \big|\Lambda_{u}(x)\,h\big(\varphi_{u}(x)\big)\,\psi_{u}(x)\big|\Big)\,du\, d\mu(x)
    < \infty,
  \]
  and
  \[
    \int_{\mathbb{D}}
    \Big(\int_0^T \Lambda_{u}(x)^2\,
    \big\|(\upsilon_u^\mu)^\intercal\mathrm{diag}\big( b(\varphi_{u}(x)) \big)\big\|^2\,du\Big)^{1/2}\,d\mu(x) < \infty.
  \]
\end{assumption}

\begin{proposition}[Aggregate semimartingale: continuous population]
  \label{prop:cont_app}
  Under Assumption~\ref{ass:SF_app}, the population-weighted aggregate
  \[
    \psi^\mu_{t}
    \coloneqq  \int_{\mathbb{D}} \psi_{t}(x)\,\Lambda_{t}(x)\,d\mu(x)
  \]
  is a continuous semimartingale for $0\leq t<\infty$.

  Moreover,
  \begin{align*}
    \psi^\mu_{t}
     & = \psi^\mu_{0}
    + \int_0^t \int_{\mathbb{D}}\Big[
      \Lambda_{u}(x)\,\zeta_u^\mu a\big(\varphi_{u}(x)\big)
      + \Lambda_{u}(x)\,h\big(\varphi_{u}(x)\big)\,\psi_{u}(x)
    \Big]\,d\mu(x)\,du                    \\
     & \quad + \int_0^t \int_{\mathbb{D}}
    \Lambda_{u}(x)\,(\upsilon_u^\mu)^\intercal\mathrm{diag}\big( b(\varphi_{u}(x)) \big)\,d\mu(x)\,dW(u).
  \end{align*}
\end{proposition}

\begin{proof}
  Apply the product rule to $t\mapsto\psi_{t}(x)\Lambda_{t}(x)$ to identify
  its drift and diffusion from Assumption~\ref{ass:SF_app}. Then integrate in $x$
  and justify exchanging the $x$-integral with the Lebesgue/stochastic integrals
  by the same dominated-convergence and stochastic-Fubini argument used in
  Proposition~\ref{prop:semimartingaleProperty}, proving the stated
  decomposition.
\end{proof}

\begin{corollary}[Aggregate semimartingale and restart property]
  \label{cor:aggregate_semimartingale_app}
  Assume Assumption~\ref{ass:discrete_app} or Assumption~\ref{ass:SF_app}, and
  adopt the population structure of Lemma~\ref{lem:aggregation_cocycle_app}.
  Fix $t\ge 0$ and set
  \[
    \mu_t \coloneqq \mu_{0,t}.
  \]

  \smallskip
  \noindent\textbf{(Restarted aggregate process).}
  For $u\ge 0$, define the restarted aggregate by
  \[
    \psi^{\mu_t}_{0,u}
    \coloneqq
    \int_{\mathbb{D}} \bar\psi_{t,t+u}(x)\,\Lambda_{t,t+u}(x)\,d\mu_t(x)
    \;=\;
    \int_{\mathbb{D}} \psi_{t,t+u}(x)\,d\mu_t(x).
  \]
  Then $u\mapsto \psi^{\mu_t}_{0,u}$ is a continuous semimartingale on $[0,\infty)$.

  Moreover, its drift and diffusion coefficients are the corresponding
  population-weighted aggregates of the individual coefficients computed with
  respect to the shifted initial population $\mu_t$ and the shifted flow $\varphi_{t,t+u}$, $u\ge 0$.

  \smallskip
  \noindent\textbf{(Pathwise identification with calendar time).}
  If $\psi$ is $\varphi$-consistent, then for every $u\ge 0$,
  \[
    \psi^{\mu_t}_{0,u}
    \;=\;
    \psi^\mu_{0,t+u}
    \;=:\;
    \psi^\mu_{t+u}.
  \]
  In particular, the restarted process $u\mapsto \psi^{\mu_t}_{0,u}$ coincides
  pathwise with the calendar-time process $u\mapsto \psi^\mu_{t+u}$.
\end{corollary}
\bibliographystyle{plainnat}

\end{document}